\documentclass[aps,prl,twocolumn,showpacs,superscriptaddress,groupedaddress,floatfix,10pt]{revtex4-2}
\pdfoutput=1
\usepackage[utf8]{inputenc}
\usepackage[english]{babel}
\usepackage{amsmath,amssymb}
\usepackage{bm,bbm}
\usepackage{mathrsfs}
\usepackage{graphicx}
\usepackage{color}
\usepackage{mathtools}
\usepackage{amsfonts}
\usepackage{booktabs}
\usepackage{amsthm}
\usepackage{url}

\newcommand{\e}[1]{\mathrm{e}^{#1}}
\newcommand{\ee}{\mathrm{e}}
\newcommand{\eeq}{\mathrm{eq}}
\newcommand{\TT}{\tilde{T}}
\newcommand{\LL}{\hat{\mathcal{L}}}

\newcommand{\bx}{\mathbf{x}}
\newcommand{\by}{\mathbf{y}}
\newcommand{\bxx}{\mathbf{X}}
\newcommand{\br}{\mathbf{r}}
\newcommand{\brr}{\mathbf{R}}
\newcommand{\bq}{\mathbf{q}}
\newcommand{\BB}{\mathrm{B}}
\newcommand{\BA}{\mathbf{A}}

\newcommand{\bgam}{\boldsymbol{\Gamma}}
\newcommand{\bxi}{\boldsymbol{\Xi}}
\newcommand{\mpp}{\mathcal{P}}

\newcommand{\mD}{\mathcal{D}}
\newcommand{\mL}{\mathcal{L}}
\newcommand{\mU}{\mathcal{U}}
\newcommand{\mS}{\mathcal{S}}
\newtheorem{thm}{Theorem}

  \newcommand{\sump}{{\sum_{k}}'}
\newcommand{\prodp}{{\prod_{k}}'}

\begin{document}
\title{Faster uphill relaxation in thermodynamically equidistant temperature quenches}
\author{Alessio Lapolla}
\author{Alja\v{z} Godec}
\email{agodec@mpibpc.mpg.de}
\affiliation{Mathematical bioPhysics group, Max Planck Institute for
  Biophysical Chemistry, G\"{o}ttingen 37077, Germany}
  \begin{abstract}
We uncover an unforeseen asymmetry in relaxation -- for
a pair of thermodynamically equidistant temperature quenches, one from
a lower and the other from a higher temperature, the relaxation
at the ambient temperature is faster in case of the
former. We demonstrate this finding on hand of two exactly solvable 
many-body systems relevant in the context of single-molecule and
tracer-particle dynamics. We prove that near stable minima and for
all quadratic energy landscapes it is a general phenomenon that also
exists in a class of
non-Markovian observables probed in single-molecule and
particle-tracking experiments. The asymmetry is a
general feature of reversible overdamped diffusive systems with
smooth single-well potentials and occurs in multi-well landscapes when
quenches disturb predominantly intra-well equilibria. Our findings may be relevant for the optimization
of stochastic heat engines. 
  \end{abstract}
\maketitle

 Relaxation processes are a paradigm for condensed matter
 \cite{farhan_direct_2013,dattagupta_relaxation_2012}, single-molecule experiments \cite{Chen_2007}
 and tracer-particle transport in complex media \cite{Wang_2006,lizana_single-file_2008,lizana_diffusion_2009,Lapolla_2018,lapolla_manifestations_2019}. Relaxation close to equilibrium
 was described by the mechanical Onsager-Casimir
 \cite{onsager_reciprocal_1931,onsager_reciprocal_1931-1} and
thermal Kubo-Yokota-Nakajima \cite{Kubo_1957} linear laws.
 These pioneering ideas were consistently generalized in numerous ways, most notably,
 to thermodynamics along individual stochastic trajectories driven far
 from equilibrium at weak \cite{seifert_stochastic_2012,
   jarzynski_equalities_2011} and strong \cite{Seifert_strong,Strasberg_2016,Massi,Jarzynski_strong,talkner2019}
 coupling with the bath, anomalous diffusion phenomena \cite{Ralf_1,Ralf_2,Ralf_3,Igor},
 and the so-called 'frenesis' focusing on the dynamical
 activity -- a dynamic counterpart to changes in entropy \cite{Baiesi_2013,Maes_2017}. 
Many of these new concepts have been verified by and/or successfully applied
in experiments in colloidal systems \cite{Blickle_2006,Hoang_2018,Jeon_2013} and single-molecule
experiments on nucleic acids 
\cite{Collin_2005,Dieterich_2015,Camunas_Soler_2017} and larger
biomolecular machines \cite{Hayashi_2010}. 
 
 Not as much is known about transients, in particular those evolving from
   non-stationary initial conditions. Our present understanding
 of thermodynamics and in particular the kinetics in transient systems, reversible as well as irreversible, is mostly limited to small
 deviations from equilibrium
 \cite{onsager_reciprocal_1931,onsager_reciprocal_1931-1},
 non-equilibrium steady states
 \cite{Sasa,Maes_2011,Baiesi_2013,polettini_nonconvexity_2013,Maes_2019},
 and statistics of the 'house-keeping' heat \cite{Speck,Roldan} and
 entropy production \cite{Noh}. The
 r\^ole of initial conditions in relaxation was recently
 studied in the context of the 'Mpemba effect' -- the phenomenon where
 a hot system can cool down faster than the same system initiated at a lower temperature   \cite{lu_nonequilibrium_2017,klich_mpemba_2019}. Notable recent
advances include an information-theoretic bound on the entropy production
during relaxation far from equilibrium
\cite{shiraishi_information-theoretical_2019} and a spectral duality
between relaxation and first passage processes \cite{David,Hartich_2019}. 

It is meanwhile possible to probe the transient, non-equilibrium
dynamics of colloids and single molecules, e.g.
by temperature-modulated particle tracking \cite{Wang_2006},
time-\cite{Edgar_NP} and temperature-modulated \cite{Edgar_PRE}, 
temperature-jump \cite{delorenzo_2015} and holographic
\cite{Gladrow_2019} optical tweezers as well as optical pushing
\cite{pushing}.
These experiments allow for systematic investigations of the dependence of relaxation on the direction of
the displacement from equilibrium, which is the central question of the
present Letter.

Notwithstanding, the dependence of relaxation on the direction of the displacement from
equilibrium (see Fig.\ref{thermo:fig}) remains elusive. 
 Moreover, as a result of the projection to a
 lower-dimensional subspace it is expected that
 observables in many experiments, in particular those tracking
 individual particles  
 \cite{Wang_2006} and single-molecules \cite{delorenzo_2015,Gladrow_2019}, relax in a manner that is not Markovian \cite{lapolla_manifestations_2019}.
\begin{figure}
    \includegraphics[width=0.48\textwidth]{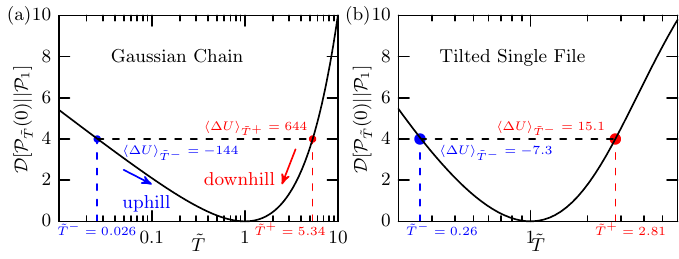}
    \caption{
      Non-equilibrium free energy
      after a temperature quench $T\to T_{\eeq}$ at time $t=0$ in units
      of $k_{\mathrm{B}}T_{\eeq}$,
      $\Delta
  F_{\TT}=\mathcal{D}[\mpp_{\TT}(t=0^+)||\mpp^\eeq_1]$ (see Eq.~(\ref{KLD})), as a function of the relative pre-quench
      temperature $\tilde{T}=T/T_{\eeq}$ (note the logarithmic
      scale); a) refers to the
      end-to-end distance of a Gaussian chain with 100
      beads and b) to the 7-th  in a single file of 10 particles in a
      linear potential with slope 10 confined
      to a unit box. The blue and red points depict
      a pair of thermodynamically equidistant temperature quenches, $\tilde{T}^-,\tilde{T}^+$,  with corresponding
      excess potential energies  
      $\langle \Delta U\rangle_{\TT^{\pm}} \equiv \langle
      U(0^+)\rangle_{\TT^{\pm}}-\langle U\rangle_1$.}
    \label{thermo:fig}
  \end{figure}

Here, we address relaxation from
an instantaneous temperature quench $T\to T_{\eeq}$ at time $t=0$ with respect to its directionality,
$T^-\uparrow T_\eeq$ versus $T^+\downarrow T_\eeq$. We uncover an unforeseen
 dependence on the direction of the quench --
 for a given pair of temperatures $T^-<T_{\eeq}<T^+$ at which the
   thermodynamic displacement from equilibrium at $t=0^+$ in the sense
   of $\mathcal{D}_{T^{\pm}}(0^+)$ -- the non-equilibrium free energy
   difference or 'lag'
   \cite{s._kullback_information_1951,lebowitz_irreversible_1957,Mackey_1989,Qian_2013,Massi_proof,Massi_PRL,Vaiku}--
   is equal, i.e. $\mathcal{D}_{T^+}(0^+)=\mathcal{D}_{T^-}(0^+)$  (see Fig.~\ref{thermo:fig}),
 relaxation evolves, contrary to intuition, faster 'uphill' ($\langle \Delta U\rangle_{T^{-}}<0$)
 than 'downhill'  ($\langle \Delta U\rangle_{T^{+}}>0$)
 the energy landscape.
 This always holds for single-well potentials and occurs in
 near degenerate multi-well
 potentials with high energy barriers, under Markovian
 dynamics as well as for a class of non-Markovian observables
 probed by single-molecule and particle-tracking experiments. 
 We demonstrate the asymmetry on hand of the
 Gaussian polymer chain
 \cite{doi_theory_1988}, single-file diffusion in a tilted box
 \cite{lapolla_manifestations_2019} and for diffusion in nearly degenerate
 multi-well potentials. For relaxation near a stable minimum and thus for all reversible Ornstein-Uhlenbeck processes we
 prove that the asymmetry, albeit counter-intuitive, is general.

\emph{Theory.---} We consider $d$-dimensional Markovian diffusion
with a $d\times d$ symmetric positive-definite diffusion matrix
$\mathbf{D}$ and mobility tensor
$\mathbf{M}_T=\mathbf{D}/k_{\mathrm{B}}T$ in a drift field
$\mathbf{F}(\bx)$ such that $\mathbf{M}_T^{-1}\mathbf{F}(\bx)=-\nabla U(\bx)$
is a gradient flow. The evolution of
the probability density at temperature $T$
is governed by the Fokker-Planck operator $\hat{\mL}_T\equiv\nabla\cdot \mathbf{D}\nabla-\nabla\cdot
\mathbf{M}_T\mathbf{F}(\mathbf{x})$. We let $G_T(\mathbf{x},t|\mathbf{x}_0)$ be the Green's function of the
initial value problem
$(\partial_t-\hat{\mL}_{T})G_T(\mathbf{x},t|\mathbf{x}_0)=0$, 
and assume that the potential
$U(\mathbf{x})$ is confining (i.e. $\lim_{|\mathbf{x}|\to\infty}U(\mathbf{x})=\infty$).
This assures the existence
of an invariant Maxwell-Boltzmann measure with density
$\lim_{t\to\infty}G_T(\mathbf{x},t|\bx_0)\equiv
P_{T}^{\eeq}(\mathbf{x})=Q_T^{-1}\mathrm{e}^{-U(\mathbf{x})/k_{\mathrm{B}}T},\forall
\bx_0$ with partition function $Q_T=\int
\mathrm{e}^{-U(\mathbf{x})/k_{\mathrm{B}}T}d\mathbf{x}$.

The system is prepared at equilibrium with a
temperature $T$, $P_{T}^{\mathrm{inv}}(\mathbf{x})$, whereupon an instantaneous temperature quench
is performed to the ambient temperature $T_{\eeq}$ at $t=0$.
The relaxation
evolves at  $T_{\mathrm{eq}}$ according to $\hat{\mL}_{T_\eeq}$
and for a given system it is uniquely characterized by $T$. 
For convenience we define $\tilde{T}\equiv
T/T_{\mathrm{eq}}$ \footnote{$T>T_{\eeq}$ implies $\TT>1$ and
  $T<T_{\eeq}$ implies $0<\TT<1$.}, such that
  \begin{equation}
   P_{\TT}(\bx,t)=\!\int\!\! d\bx_0
   G_1(\bx,t|\bx_0)P_{\TT}^{\eeq}(\bx_0)
   \xrightarrow[t\to\infty]{} P_{1}^{\eeq}(\bx).
      \label{density}
  \end{equation}
The instantaneous entropy and mean energy are given by
$S_{\TT}(t)\equiv-k_{\BB}\int d\mathbf{x} P_{\TT}(\mathbf{x},t)\ln P_{\TT}(\mathbf{x},t)$ and $\langle U(t)\rangle_{\TT}=\int d\mathbf{x} P_{\TT}(\mathbf{x},t)
U(\mathbf{x})$, respectively, where $\langle\cdot\rangle_{\TT}$ denotes an
average over all paths $\bx(t)$ starting from $P_{\TT}^{\mathrm{inv}}(\bx_0)$.

Let the measured physical observable be
$\bq=\bgam(\bx)$. Its probability density function corresponds to \cite{lapolla_manifestations_2019}
 \begin{equation}
 \mpp_{\TT}(\bq,t)=\hat{\Pi}_{\mathbf{x}}(\mathbf{q})P_{\TT}(\bx,t)\!\equiv\!\!\int\!\!
 d\bx\delta(\bgam(\bx)-\bq)P_{\TT}(\bx,t),
\label{nMark} 
  \end{equation}
which in general displays non-Markovian dynamics as soon as
$\bq$ corresponds to a low-dimensional projection
\cite{lapolla_manifestations_2019}. Once equilibrium is reached
 we have
 $\lim_{t\to\infty}\mathcal{P}_{\TT}(\bq,t)=\mathcal{P}^{\mathrm{eq}}_{1}(\bq)$,
or, expressed via the so-called potential of mean
 force $\mU(\bq)$ \footnote{The name comes from the fact that
   $\mU(\bq)$ delivers the mean force, i.e. \unexpanded{$-\nabla_{\bq}
     \mU(\bq)=-\langle \nabla_{\bx}U(\bx)\delta(\bgam(\bx)-\bq)\rangle$}.}, $\mathcal{P}^\mathrm{eq}_1(\bq)=\ee^{-\beta_\eeq \mU(\bq)}$ \cite{Kirkwood,Seifert_strong,Jarzynski_strong}. Obviously,
when $\bgam(\bx)=\bx$ we have $\mpp_{\TT}(\bq,t)=P_{\TT}(\bx,t)$.

We quantify the instantaneous displacement from equilibrium
with the Kullback-Leibler divergence \cite{s._kullback_information_1951,lebowitz_irreversible_1957,Mackey_1989,Qian_2013,Massi_proof,Massi_PRL,Vaiku}
\begin{equation}
  \mathcal{D}[\mpp_{\TT}(t)||\mpp^{\eeq}_1]=\int d\bq
    \mpp_{\TT}(\bq,t)\ln(\mpp_{\tilde{T}}(\bq,t)/\mpp^{\eeq}_1(\bq)).
   \label{KLD}
\end{equation}
Writing this out for the Markovian case
we find, upon identifying $S_{\TT}(t)$ and $\langle U_{\TT}
 (t)\rangle$ 
\begin{equation}
\mathcal{D}[P_{\TT}(t)||P^{\eeq}_1]=-S_{\TT}(t)/k_{\mathrm{B}}+\beta_{\eeq}\langle U(t)\rangle_{\TT}
 + \ln  Q_{T_\eeq}
\label{KLDM}.
\end{equation}
Recalling the definition of free energy
$F=-\beta_{\eeq}^{-1}\ln Q_{T_\eeq}$ and 
defining the instantaneous generalized free energy (GFE)
\cite{Qian_2013} or 'lag' \cite{Vaiku} as
  $F_{\TT}(t)=\langle U(t)\rangle_{\TT} -T_{\eeq}S_{\TT}(t)$
 we see, upon multiplying through by $\beta_{\eeq}^{-1}=k_{\mathrm{B}}T_\eeq$, that in the
  Markovian case Eq.~(\ref{KLD}) is the
  excess GFE in units of $k_{\mathrm{B}}T_{\eeq}$, i.e. $\mathcal{D}_{\TT}^M(t)\equiv\mathcal{D}[P_{\TT}(t)||P_1^{\eeq}]=\beta_{\eeq}(F_{\TT}(t)-F)$ \cite{Mackey_1989,Qian_2013}. 
Writing out Eq.~(\ref{KLD}) for the
non-Markovian case and identifying $\mS_{\TT}(t)$ and $\mU(\bq)$
(calligraphic letters denote
potentials of projected observables)
we find
\begin{equation}
\mathcal{D}_{\TT}^{nM}(t)\equiv\mathcal{D}[\mpp_{\TT}(t)||\mpp_1^{\eeq}]=\!\!-\mathcal{S}_{\TT}(t)/k_{\mathrm{B}}+\beta_\eeq\langle \mU(t)\rangle_{\TT}
\label{KLDMN},
\end{equation}
 which is the non-Markovian GFE,
   $\mathcal{D}_{\TT}^{nM}(t)=\beta_{\eeq}\mathcal{F}_{\TT}(t)$. Note
  that  $\mU(\bq)$ itself is an effective free energy,
    i.e. $\beta_\eeq\mU(\bq)\equiv-\ln\langle\delta(\bgam(\bx)-\bq)\rangle_{1}=-\ln\int d\bx\delta(\bgam(\bx)-\bq)\ee^{-\beta_\eeq
    U(\bx)}+\ln Q_{T_\eeq}$ and $\mS_1=-\langle\mU\rangle_1$.
  We henceforth express energies in
  units of $k_{\mathrm{B}}T_{\mathrm{eq}}$.  
If (and only if) latent degrees of freedom (i.e. those integrated out) relax much faster than
$\bq(t)$,
Eqs.~(\ref{KLDM}) and (\ref{KLDMN}) are equivalent and $\bq(t)$ is a Markovian diffusion in the free energy
landscape $\mU(\bq)$  \cite{lapolla_manifestations_2019}. In absence
of a time-scale separation, however, both $\mS_{\TT}(t)$ and
$\langle\mU(t)\rangle_{\TT}$ contain contributions from the (hidden)
relaxation of the latent degrees of freedom.

  Consider now a pair of temperatures $\TT^+>1$ and $\TT^-<1$ corresponding
  to equal displacements immediately after the quench: $\mathcal{D}^{M,nM}_{\TT^-}(0^+)=\mathcal{D}^{M,nM}_{\TT^+}(0^+)$. The
  existence of (at least) two such temperatures is guaranteed within
  an interval $\TT\in (\TT_{\mathrm{min}},\TT_{\mathrm{max}})$ where
  $\mathcal{D}^{M,nM}_{\TT}(0^-)=f(\TT)$ has no local maximum. The central
  question of this Letter addresses the rate of the 'uphill'
  ($\TT^-<1$) versus 'downhill' ($\TT^+>1$) relaxation.
  
\emph{Gaussian Chain.---} In the context of single-molecule
experiments we consider the overdamped dynamics of a chain of $N+1$ beads with coordinates $\{\br_i\}$ connected by harmonic
springs with potential $U(\{\br_i\})=\sum_{i=1}^N(\br_{i+1}-\br_i)^2$
(general harmonic networks are treated in the SM). In the Markovian setting we
consider all monomers, $P_{\TT}(\{\br_i\},t)$ in Eq.~(\ref{density}), while single-molecule experiments (e.g. FRET \cite{yang_single-molecule_2002,joo_advances_2008} or
optical tweezers  \cite{delorenzo_2015,Gladrow_2019})  typically
track a single (e.g end-to-end) distance within the macromolecule,
$\bq\equiv d=|\br_1-\br_N|$ with $\mathcal{P}_{\TT}(d,t)$ from Eq.~(\ref{nMark}), evolving
according to non-Markovian dynamics. 
  \begin{figure*}[ht!!]
  \begin{center}
   \includegraphics[width=0.98\textwidth]{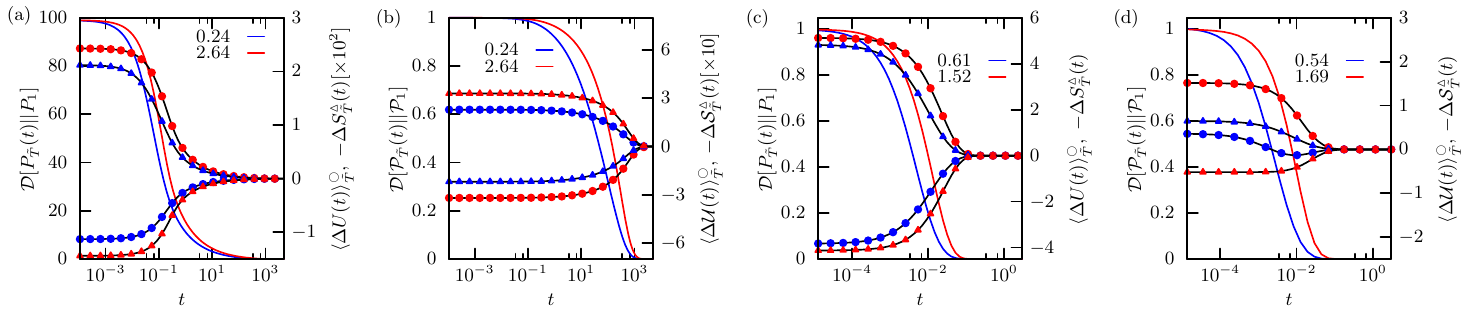}
   \caption{
      $\mathcal{D}[\mpp_{\TT}(t)||\mpp^{\eeq}_1]$ (full lines) for the Gaussian chain (a and b) and
     single-file with 10 particles in a linear potential with slope $g=10$ (s and d). 
     a) refers to the entire chain of 100 beads (Eq.~(\ref{KLDM})) and
     b) to the end-to-end distance (Eq.~(\ref{KLDMN})) for equidistant quenches from $\TT^-=0.24$
     (blue) and $\TT^+=2.64$ (red); c) stands for the full single file
     for equidistant quenches from $\TT^-=0.61$
     (blue) and $\TT^+=1.52$ (red);  d) the 7-th particle for equidistant quenches from $\TT^-=0.54$
     (blue) and $\TT^+=1.69$. 
     The circles refer to
     $\langle \Delta U(t) \rangle_{\TT^{\pm}}$ and $\langle\Delta \mU(t)
     \rangle_{\TT^{\pm}}$ in a and c, and b and d, respectively, and triangles
     denote $\Delta S_{\TT^{\pm}}(t)$ and $\Delta \mS_{\TT}(t)$. Note the
     second axes for $\langle \Delta U(t) \rangle_{\TT^{\pm}},\langle \Delta\mU(t)
     \rangle_{\TT^{\pm}}$ and $\Delta
     S_{\TT^{\pm}}(t),\Delta\mS_{\TT}(t)$. Note that $\mS_{\TT}(\infty)=\mS_1=-\langle\mU\rangle_1$.} 
   \label{energydecay}
   \end{center}
  \end{figure*}
  
The excess GFE is given by
(see derivation in the SM)
\begin{eqnarray}
\label{KLRouse}   
\mathcal{D}^M_{\TT}(t)&=&\frac{3}{2}\sum_{k=1}^N\left[\Lambda_k^{\TT}(t)-1-\ln
  \Lambda_k^{\TT}(t)\right]\\
\mathcal{D}_{\TT}^{nM}(t)&=&\frac{3}{2}\left[\frac{\mathcal{A}_{\TT}^{1N}(t)}{\mathcal{A}_{1}^{1N}(0)}-1-\ln\frac{\mathcal{A}_{\TT}^{1N}(t)}{\mathcal{A}^{1N}_1(0)}\right],
\label{KLRouseN}  
\end{eqnarray}
where $\Lambda_k^{\TT}(t)\equiv 1+(\TT-1)\e{-2\mu_kt}$ with $\mu_k=4\sin^2(\frac{k\pi}{2(N+1)})$ and
we introduced $\mathcal{A}_{\TT}^{ij}(t)\equiv\sum_{k=1}^N\Lambda_k^{\TT}(t)\mathcal{C}_k^{ij}/2\mu_k$
with $\mathcal{C}_k^{ij}\ge 0$ given explicitly in the SM. The initial
excess free energies are both
convex in $\TT$ and read
\begin{equation}
\mathcal{D}_{\TT}^M(0^+)=3N(\TT-1-\ln\TT)/2=N\mathcal{D}_{\TT}^{nM}(0^+).
\label{KLRouse_I}  
\end{equation}
The instantaneous potential
  energy of the full system and the
  potential of mean force in turn read
$\langle U(t)  \rangle_{\TT}=\frac{3}{2}\sum_{k=1}^{N}\Lambda_k^{\TT}(t)$ and 
  $\mU(d)=-\ln \mpp_1^{\mathrm{eq}}(d)$, respectively.  
  Aside from specific values of $\mu_k$ and $\mathcal{C}^{ij}_k$
  Eqs.~(\ref{KLRouse}-\ref{KLRouse_I}) hold for any reversible
  Ornstein-Uhlenbeck process, that is for
  any $\TT$, connectivity/topology, and tagged distance. 

  The results
  for $\mathcal{D}_{\TT}^{M,nM}(t)$ and their decomposition into $\langle U\rangle_{\TT},\langle \mU(t)\rangle_{\TT},S_{\TT}(t),$
  and $\mS_{\TT}(t)$ for a pair of equidistant temperature quenches
 are shown in Fig.~\ref{energydecay} and demonstrate that the uphill
 relaxation is always faster than downhill relaxation.
As we prove below
this is true for any reversible Ornstein-Uhlenbeck process (OUp)
quenched arbitrarily far from equilibrium.

The energy and entropy differences relative to their
  equilibrium values (i.e. at $t=\infty$) in Fig.~\ref{energydecay}a
  suggest that the Markovian uphill and downhill relaxation are
dominated by $\langle \Delta U(t)
\rangle_{\TT^{+}}$ and $\Delta S_{\TT^{-}}$, respectively.
Surprisingly, entropy pushing the
system uphill against the deterministic force is more
efficient. Notably, the magnitude of individual 
  contributions is smaller for uphill relaxation,
  i.e. $\langle \Delta U
\rangle_{\TT^{+}}>-\langle \Delta U
\rangle_{\TT^{-}}$ and $\Delta S_{\TT^{+}}>-\Delta S_{\TT^{-}}$. Thus,
a larger energy excess and entropy deficit are dissipated
during downhill relaxation. Conversely, the partitioning into $\mS_{\TT}(t)$ and
  $\langle\mU(t)\rangle_{\TT}$ of the non-Markovian relaxation depends
  on the details of the projection and is less intuitive (in our
  example in Fig.~\ref{energydecay}b it is in fact reversed).
  \begin{figure}[ht!!]
   \includegraphics[width=8.5cm]{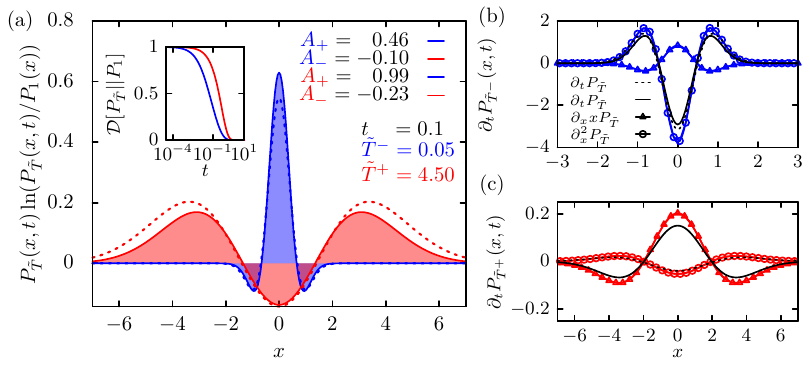}
   \caption{a) Integrand of Eq.~(\ref{KLD}) at $t=0.1$
     for a
     one-dimensional OUp (full line)
     for uphill (blue) and downhill (red) relaxation with the positive $A_+$ and negative $A_-$ area
     under the curve; Inset: The corresponding
     $\mathcal{D}_{\TT}^M(t)$. b-c) Decomposition of $\partial_tP_{\TT}(x,t)$ into 
diffusive $\partial^2_xP_{\TT}$ (circles) and advective $\partial_xx P_{\TT}$
     (triangles) contribution for uphill (b) and downhill (c)
relaxation. Dashed lines correspond to free
     diffusion evolving from the same initial condition.}
   \label{explain}
  \end{figure}

To explain why uphill relaxation is faster we
  inspect in Fig.~\ref{explain} local contributions to
  $\mathcal{D}_{\TT}^M(t)$ for a one-dimensional OUp.
  An uphill quench localizes $P_{\TT^-}(x,0^+)$ near the origin, whereas a
downhill quench broadens $P_{\TT^+}(x,0^+)$ rendering the
integrand of Eq.~(\ref{KLDM}) non-zero over a larger domain 
(Fig.~\ref{explain}a, red line).
The evolution of
$P_{\TT}(x,t)$ is driven by diffusion
$\propto \partial^2_xP_{\TT}$ and advection $\propto
\partial_xx P_{\TT}$. By forcing probability mass towards the origin
advection seems to oppose uphill relaxation (triangles in Fig.~\ref{explain}b) but thereby actually
sustains an even faster diffusion rate compared to free diffusion
(compare circles and dashed line in Fig.~\ref{explain}b). The net effect
is an overall relaxation nearly as fast as free diffusion (compare full
and dashed line in Fig.~\ref{explain}b). Downhill relaxation is
advection-dominated and weakly opposed by
diffusion, which is almost unaffected by the potential
(Fig.~\ref{explain}c). The overall dynamics is much slower (compare
full lines in Fig.~\ref{explain}b \& c). Faster diffusion from a
localized initial distribution thereby renders uphill relaxation
faster -- an effect that will exist in any confining potential well
with ruggedness $\ll k_{\rm B}T_{\eeq}$. $\hat{\mL}_T$ of any reversible OUp is diagonalizable and thus
uniquely decomposable into one-dimensional OUps, 
extending our explanation to arbitrary dimensions.
  
Non-Markovian relaxation displays the same asymmetry
  but the dominant driving forces, here $\mS_{\TT}$ and
  $\langle\mU_{\TT}\rangle$, 
may become reversed (see
Fig.~\ref{energydecay}b). Since $\mU_{\TT}$ contains entropic effects
of latent degrees of freedom the partitioning between $\mS_{\TT}$ and
$\langle\mU_{\TT}\rangle$ is in general projection-dependent.

  \emph{Tilted Single File.---} In the context of tracer-particle
  dynamics we consider $N$ hard-core Brownian point particles with positions $\{x_i(t)\}$ (the extension to a finite
  diameter is straightforward \cite{lizana_diffusion_2009,
    lizana_single-file_2008}) diffusing in a box of unit
  length in the presence of a linear potential (e.g. the gravitational
  field),
  $U(\{x_i\})=\sum_{i=1}^N g x_i$. The probability density
    of $\{x_i(t)\}$  upon a
    quench from $\TT$, $P_{\TT}(\{x_i(t)\},t)$, evolves according to 
  $\hat{\mathcal{L}}_1=\sum_{i=1}^N(\partial_{x_i}^2+g\partial_{x_i})$ under non-crossing conditions
  \cite{Lapolla_2018,lapolla_manifestations_2019}.
  In the SM we solve the problem exactly
  via the coordinate Bethe ansatz
  \cite{Lapolla_2018,lapolla_manifestations_2019}, both for the
    Markovian complete single file and the non-Markovian probability density of a tagged
  particle $\mathcal{P}_{\TT}(z,t)$ (i.e. $\bq\equiv x_{\mathcal{T}}=z$).

  $\mathcal{D}^M_{\TT}(t)$ with
  corresponding $\langle \Delta U(t)\rangle_{\TT},\langle
  \mU(t)\rangle_{\TT},\Delta S_{\TT}(t)$ and $\mS_{\TT}(t)$
  for the complete and tagged particle dynamics
  are shown in
  Fig.~\ref{energydecay} c and d. As for the Gaussian
  chain uphill relaxation in the tilted single file, both full as well
  as for a tagged particle, seems to always be faster, irrespective of which particle we tag and for
  any $\TT, N$ and tilting strength $g> 0$ (see also SM).
The Markovian uphill
relaxation is dominated by $\Delta S_{\TT^-}(t)$
and downhill by $\langle \Delta U(t)\rangle_{\TT^+}$,  and a larger energy and entropy difference
must be dissipated during downhill relaxation (see
Fig.~\ref{energydecay}c). For a tagged-particle
  the partitioning between $\langle
  \mU(t)\rangle_{\TT^-}$ and $\mS(t)_{\TT^-}$ varies
  depending on which particle we tag as a result of the shape of
  $\mU(z)$ and the dependence of $\mpp_{\TT}(z,0^-)$ on $\TT$, which in turn both depend on the tagged particle as
  well as $T_\eeq, N$ and $g$.

  \emph{Is the asymmetry universal?---} We first focus on dynamics
  near a stable minimum at $\brr_0$, $\delta\brr(t)=\brr(t)-\brr_0$, which is well described by an
  Ornstein-Uhlenbeck process (OUp), i.e. $d\delta
  \brr(t)=\mathrm{\mathbf{H}}\delta\brr(t) dt +
  \sqrt{2}d\mathbf{W}_t$, where $(\mathrm{\mathbf{H}})_{ij}=\sum_{ij}\partial_{R_i}\partial_{R_j}U(\brr)|_{\brr_0}$ is the
  Hessian. 
\begin{thm}
  For a general diffusion sufficiently close to
  a stable minimum and for any stable reversible OUp the relaxation from
  a pair of equidistant quenches of arbitrary magnitude (as defined above) is
  always faster uphill.
\end{thm}
\begin{proof}
\textcolor{red}{\textbf{(See Erratum)}} Any pair $0<\TT^-\le1$ and
    $1\le \TT^+<\infty$ with $\mathcal{D}_{\TT^+}(0^+)=\mathcal{D}_{\TT^-}(0^+)$ satisfies by construction   
  $\TT^+-\TT^-=\ln(\TT^+/\TT^-)$. 
  We first prove the claim for the Markovian setting, where
  Eq.~(\ref{KLRouse}) has the structure
  $\mD_{\TT}^M(t)=\sum_{k=1}^N\mD^{\pm}_{k}(t)$. 
  We set
  $\varphi\equiv\TT^+/\TT^->1$, $\delta_{\pm}\equiv
  \TT^{\pm}-1$, 
  and
  write
  $\Delta\mD_k(t)\equiv\mD^{+}_k(t)-\mD^{-}_k(t)=\ln Z_{\varphi}(\mu_kt)$,
  such that
  \begin{eqnarray}
    Z_{\varphi}(\tau)&=&\varphi^{\ee^{-\tau}}(1+\delta_{-}\ee^{-\tau})/(1+\delta_{+}\ee^{-\tau})\nonumber\\
    &=&[
   \varphi (1+\delta_{-}\ee^{-\tau})^{\ee^{\tau}}]^{\ee^{-\tau}}/(1+\delta_{+}\ee^{-\tau})\nonumber\\
    &\ge& [\varphi(1+\delta_{-})/(1+\delta_{+})]^{\ee^{-\tau}}\ge 1,
  \end{eqnarray}
 where we have used both generalized Bernoulli inequalities, i.e.
 for any real $0\le y_{-}\le 1$, $y_{+}\ge 1$ and $x\ge -1$ we have
 $(1+x)^{y_+}\ge 1+y_+x$ and $(1+x)^{y_-}\le 1+y_-x$. Recalling the
 definition of $\Delta\mD_k(t)$ completes the proof. 
  
  To prove the claim in the non-Markovian setting for
  projections of type $q=|\delta\brr_i-\delta\brr_j|$ we
  first realize that
$\dot{\mathcal{A}}_{\TT^+}^{ij}(t)\le0$ and
  $\dot{\mathcal{A}}_{\TT^-}^{ij}(t)\ge0$, where $\dot{f}(t)\equiv \frac{d}{dt}f(t)$. Setting $\Delta\mD(t)\equiv \mathcal{D}^{nM}_{\TT^+}(t)-\mathcal{D}^{nM}_{\TT^-}(t)$
and using Eq.(\ref{KLRouseN}) we find, upon taking the derivative
 \begin{eqnarray}
   \!\!\!\!\!\!&&\Delta\dot{\mD}(t)=
   \frac{\dot{\mathcal{A}}_{\TT^+}^{ij}(t)}{\mathcal{A}_{\TT^+}^{ij}(0)}-\frac{\dot{\mathcal{A}}_{\TT^+}^{ij}(t)}{\mathcal{A}_{\TT^+}^{ij}(t)}-\frac{\dot{\mathcal{A}}_{\TT^-}^{ij}(t)}{\mathcal{A}_{\TT^-}^{ij}(0)}+\frac{\dot{\mathcal{A}}_{\TT^-}^{ij}(t)}{\mathcal{A}_{\TT^-}^{ij}(t)}.
   \label{proofN}
 \end{eqnarray}
 Eq.~(\ref{proofN}) implies 
 $\Delta\dot{\mD}(t)\ge 0$ because
$\mathcal{A}_{\TT^+}^{ij}(t)\ge\mathcal{A}_{\TT^+}^{ij}(0)$ while
$\mathcal{A}_{\TT^-}^{ij}(t)\le\mathcal{A}_{\TT^+}^{ij}(0)$, which completes the proof.
\end{proof}
The fact that tilted single-file diffusion, being anharmonic and asymmetric with
non-perturbative interactions, displays the asymmetry for quenches of
arbitrary magnitude and for any steepness of the potential 
hints that the asymmetry might be more
general. Note that tagging different particles in different
slopes $g> 0$ we can construct $\mU(z)$ with arbitrary asymmetry.
Alongside the physical principle underlying the asymmetry
  established for the OUp
and Theorem 1 this strongly suggest that uphill relaxation in smooth single-well
potentials could be universally faster (see SM). Since the projection
(\ref{nMark}) is independent of $\TT$
these statements should extend also to non-Markovian
observables, in particular those probed in many single-molecule and
particle-tracking experiments.

As a corollary uphill relaxation is faster also in
  multi-well potentials for equidistant quenches that predominantly disturb only the intra-well
  equilibria, in particular for
  nearly degenerate basins separated by sufficiently high barriers \cite{Moro}
  (for reasoning and examples see SM). This is violated in asymmetric
  multi-wells and examples with faster downhill relaxation are constructed in the SM.

 \emph{Conclusion.---} We uncovered an unforeseen asymmetry in the
 relaxation to equilibrium in equidistant temperature
 quenches. Uphill relaxation was found to be faster -- a
 phenomenon we proved to be universal for quenches of dynamics near
 stable minima. We hypothesize that it is a general phenomenon in
 reversible overdamped diffusion in single-well potentials
 extending to degenerate multi-well potentials for quenches
   leaving inter-well equilibria virtually intact. The dependence on
 the direction of
 the quench, 
 which so far seems to have been
 overlooked, implies a systematic asymmetry in the dissipation of
 the system's entropy $\dot{S}_{\TT}(t)$ versus
 heat $\langle \dot{U}(t)\rangle_{\TT}$
 \cite{seifert_stochastic_2012} 
 and, for specific projections,
 the modified entropy
 $\dot{\mS}_{\TT}(t)$ versus 'strong
 coupling heat' $\langle \dot{\mU}(t)\rangle_{\TT}$
 \cite{Seifert_strong,Jarzynski_strong}, which seems to be relevant for the
   efficiency of stochastic heat engines
   \cite{Edgar_NP,Schmiedl,Ouerdane}. Implying that the hot isothermal step
   can be shorter than the cold one, which
   reduces
   cycle times,
   the asymmetry may also
   be relevant for the optimization of the engine's output power \cite{Edgar_NP,Schmiedl,Ouerdane}.
 Our results can readily be tested by single-molecule and
 particle-tracking experiments  \cite{Wang_2006,Edgar_NP,Edgar_PRE,delorenzo_2015,Gladrow_2019,pushing}. 
 To understand the asymmetry on the level of individual trajectories it would be
 interesting to analyze relaxation from equidistant
 quenches in terms of occupation measures
 \cite{Lapolla_2018,lapolla_manifestations_2019} and from the perspective of stochastic thermodynamics
 \cite{seifert_stochastic_2012,Seifert_strong}. 
 
  \begin{acknowledgments}
  We thank David Hartich for fruitful discussions. The financial support from the German Research Foundation (DFG) through the Emmy Noether Program GO 2762/1-1 to AG is gratefully acknowledged.
  \end{acknowledgments}


\clearpage

\onecolumngrid
\renewcommand{\theequation}{E\arabic{equation}}
\begin{center}
\textbf{Erratum:\\Faster Uphill Relaxation in Thermodynamically Equidistant Temperature Quenches}
\end{center}

While Theorem~1 in  \cite{PRL} holds true as stated, the proof of the
theorem contains a technical error, 
i.e.\ one of the inequalities in Eq.~(9) was unfortunately
applied in the false direction and there is also an error in Eq.~(10), which 
renders the proof invalid. Most importantly, the statement of Theorem 1 as well as all the
results and conclusions of the Letter \cite{PRL} remain entirely unaffected.  

Here we provide a valid proof of Theorem 1, which is slightly longer but
straightforward. The present proof covers both cases, Eqs.~(9) and
(10). We introduce
$\Gamma(t)\equiv\sum_k[\mathcal{C}_k\mathrm{e}^{-2\mu_kt}/\mu_k]/\sum_k[\mathcal{C}_k/\mu_k]$
such that $0< \Gamma(t)\le 1$,
with the notation $\delta_{\pm}\equiv\tilde T_{\pm}-1$ as in
\cite{PRL} we may
write for the Markovian, $\Delta\mathcal D_k(t)$, and
non-Markovian, $\Delta\mathcal D^{\rm nM}(t)$,  setting  
\begin{equation}
\Delta\mathcal D_k(t)=\frac32\left[(\delta_+-\delta_-)e^{-2\mu_kt}-\ln\left(\frac{1+\delta_+e^{-2\mu_kt}}{1+\delta_-e^{-2\mu_kt}}\right)\right],\qquad
\Delta\mathcal D^{\rm nM}(t)=\frac32\left[(\delta_+-\delta_-)\Gamma(t)-\ln\left(\frac{1+\delta_+\Gamma(t)}{1+\delta_-\Gamma(t)}\right)\right].
\label{first}
\end{equation}
It is to prove that $\Delta\mathcal D_k(t)\ge0$ for all $t$. By
defining $y\equiv e^{2\mu_kt}$ or 
$y\equiv\Gamma^{-1}(t)$ this is equivalent to showing 
\begin{align}
  y\ln\left(\frac{y+\delta_+}{y+\delta_-}\right)\le\delta_+-\delta_-\ .
  \label{have to show this} 
\end{align}
for all $y\ge 1$. Since $\delta_+\ge 0\ge\delta_->-1$ we have that
$x\equiv(y+\delta_+)/(y+\delta_-)\ge 1$ and we can further apply that for $x\ge
1$ we have $\ln(x)\le(x-1)/\sqrt{x}$ \cite{log} to obtain
\begin{align}
y\ln\left(\frac{y+\delta_+}{y+\delta_-}\right)\le y\frac{\frac{y+\delta_+}{y+\delta_-}-1}{\sqrt{\frac{y+\delta_+}{y+\delta_-}}}=\frac{\delta_+-\delta_-}{\sqrt{(1+\delta_+/y)(1+\delta_-/y)}},
\end{align}
which allows to conclude Eq.~\eqref{have to show this} (and thus completes
the proof) for any $y$ with $(1+\delta_+/y)(1+\delta_-/y)\ge 1$.  

For the case $(1+\delta_+/y)(1+\delta_-/y)\le 1$ we need a different
approach. Here, we write 
\begin{align}
x'&\equiv 1-\frac{(y+\delta_+)(1+\delta_-)}{(y+\delta_-)(1+\delta_+)}
=\frac{(y+\delta_-)(1+\delta_+)-(y+\delta_+)(1+\delta_-)}{(y+\delta_-)(1+\delta_+)}
=\frac{(\delta_+-\delta_-)(y-1)}{(y+\delta_-)(1+\delta_+)}\ge0,\label{x prime} 
\end{align}
($x'\ge0$ since $y+\delta_->y-1\ge 0$ and $\delta_+-\delta_-\ge 0$),
such that we can re-write the left
hand side of Eq.~\eqref{have to show this} as
\begin{align}
y\ln\left(\frac{y+\delta_+}{y+\delta_-}\right)&=y\ln\left(\frac{(y+\delta_+)(1+\delta_-)}{(y+\delta_-)(1+\delta_+)}\times\frac{1+\delta_+}{1+\delta_-}\right)
=y\ln(1-x')+y\ln\left(\frac{1+\delta_+}{1+\delta_-}\right)
=y\ln(1-x')+y(\delta_+-\delta_-),\label{log rewritten} 
\end{align}
where we used the equidistant quenches condition
$\ln\left(\frac{1+\delta_+}{1+\delta_-}\right)=\delta_+-\delta_-$. 
Recall that we are now considering $(1+\delta_+/y)(1+\delta_-/y)\le 1$ and thus $\delta_+\delta_-+y(\delta_++\delta_-)\le 0$,
and it follows (note that $\tilde{T}_-=-W_{0}(-\tilde{T}_+\mathrm{e}^{-\tilde{T}_+})$ with 
$W_0(z)$ 
denoting the principal
branch of the Lambert-W function \cite{PRL}, and $\delta_++\delta_-=\tilde T_++\tilde T_--2\ge 0$ by Theorem 3.2 in \cite{Lambert})
\begin{align}
2(y\!+\!\delta_-)(1\!+\!\delta_+)\!-\!(\delta_+\!-\!\delta_-)(y-1) 
 & =2y-(y-1)(\delta_+\!+\!\delta_-)\!+\!2[\delta_+\delta_-+y(\delta_+\!+\!\delta_-)]
\le 2y-(y-1)(\delta_++\delta_-)
\le 2y.\label{some ineq} 
\end{align}
Applying $\ln(1-x')\le\frac{-2x'}{2-x'}$ for $0\le x'<1$ \cite{log}
($x'<1$ follows from Eq.~\eqref{x prime}) we further obtain, using
Eq.~\eqref{x prime} and \eqref{some ineq} 
\begin{align}
y\ln(1-x')\le y\frac{-2x'}{2-x'}=-2y\frac{(\delta_+-\delta_-)(y-1)}{2(y+\delta_-)(1+\delta_+)-(\delta_+-\delta_-)(y-1)}
\overset{{\rm Eq.}~\eqref{some ineq}}{\le}-(\delta_+-\delta_-)(y-1),
\end{align}
noting that the expressions  negative.
Plugging this into Eq.~\eqref{log rewritten} yields Eq.~\eqref{have to show
  this} and thus completes the proof.\vspace{0.2cm}\\

\textbf{Acknowledgments.} We thank Cai Dieball for discovering the
error, and for his major contribution to finding the shortest and most elegant proof
of the theorem.

\clearpage
\newcommand{\eql}[1]{Eq.~(\ref{#1})}
\newcommand{\mxi}{\boldsymbol{\xi}}
\renewcommand{\thefigure}{S\arabic{figure}}
\renewcommand{\theequation}{S\arabic{equation}}
\setcounter{equation}{0}
\begin{center}
  \textbf{Supplementary Material for:\\
Faster uphill relaxation in thermodynamically equidistant temperature quenches}\\[0.2cm]
Alessio Lapolla and Alja\v{z} Godec\\
\emph{Mathematical bioPhysics Group, Max Planck Institute for Biophysical Chemistry, 37077 Göttingen, Germany}
\\[0.6cm]  
\end{center}

\begin{center}
\textbf{Abstract}\\[0.3cm]  
\end{center}  
\begin{quotation}
 In this Supplementary Material (SM) we present detailed derivations of
the main results for the Gaussian-Chain and tilted single-file
diffusion model presented in the main Letter, as well as several
supplementary examples with figures. We also present counterexamples
demonstrating that the uphill-downhill asymmetry is not universal as it vanishes in
sufficiently asymmetric multi-well potentials. However, we establish generic
conditions under which the asymmetry is obeyed. Finally, we also discuss the
non-Markovian Mpemba effect. 
\end{quotation}

 \section{Gaussian Chain and Ornstein-Uhlenbeck
   process}
 We consider a
 Gaussian Chain with N+1 beads with coordinates $\brr=\{\br_i\}$
 connected by harmonic springs with potential energy $U(\brr)=\frac{1}{2}\sum_{i=1}^N |\mathbf{r}_i-\mathbf{r}_{i+1}|^2$.
 The overdamped Langevin equation governing the dynamics of a
  Gaussian Chain with N+1 beads connected by ideal springs with zero
  rest-length and diffusion coefficient $D$ is given by the set of
  coupled It\^o equations
  \begin{eqnarray}
   d\mathbf{r}_1(t)&=&[-\mathbf{r}_1(t)+\mathbf{r}_2(t)]dt +\sqrt{2D}\mxi_1(t)\nonumber\\
   d\mathbf{r}_i(t)&=&[\mathbf{r}_{i-1}(t)-2\mathbf{r}_i(t)+\mathbf{r}_{i+1}(t)]dt+\sqrt{2D}\mxi_i(t)\nonumber\\
   d\mathbf{r}_{N+1}(t)&=&[-\mathbf{r}_{N+1}(t)+\mathbf{r}_N(t)]dt+\sqrt{2D}\mxi_{N+1}(t),
  \end{eqnarray}
where $\mxi_i(t)$ stands for zero mean Gaussian white noise, i.e. 
  \begin{eqnarray}
   \langle \mxi_i(t)\rangle=0, \qquad \langle \xi_{i,k}(t)\xi_{i,l}(t')\rangle=\delta_{kl}\delta(t-t').
  \end{eqnarray}
 It is straightforward to generalize these formulas to any
 reversible $M$-dimensional Ornstein-Uhlenbeck process $\brr(t)\equiv\{\br_i(t)\}$ with some $\mathbb{R}^{M}\times
 \mathbb{R}^{M}$ symmetric force matrix $\bxi$ and
 potential energy function $U(\brr)=\frac{1}{2}\brr^T\bxi\brr$
  \begin{equation}
    d\brr (t)=\bxi \brr(t) dt + \sqrt{2}d\mathbf{W}_t,
 \label{OUP}   
  \end{equation}
  where $d\mathbf{W}_t$ is the $M$-dimensional super-vector of
independent Wiener increments with zero mean and unit variance,
$\mathbb{E}[dW_{i,t}dW_{j,t'}]=\delta_{i,j}\delta(t-t')$. In this
super-vector/super-matrix notation the Gaussian chain is recovered by introducing $\mathbb{R}^{3(N+1)}\times \mathbb{R}^{3(N+1)}$ tridiagonal
super-matrix $\bxi$ with elements
\begin{equation}
\bxi_{ii}=\mathbbm{1},\quad
\bxi_{ii+1}=\bxi_{ii-1}=(-1+(-1)^{\delta_{i,1}+\delta_{i,N+1}})\mathbbm{1},
\label{matrix}
  \end{equation}
where $\mathbbm{1}$ is the $3\times 3$
identity matrix. This leads to the equations of motion presented in
the Letter. Since $\bxi$ is supposed to be symmetric these equations can be decoupled by diagonalizing $\bxi$ i.e. by
passing to normal coordinates $\brr\to\bxx\equiv\{\bx_i\}$:
\begin{equation}
\BA^T\bxi\BA=\mathrm{diag}(\boldsymbol{\mu})
\end{equation}
where the diagonal matrix has elements $\mathrm{diag}(\boldsymbol{\mu})_{kk}=\mu_k$.
This yields eigenvalues $\mu_i$ and orthogonal super-matrices $(\BA)_{ij}$, where the $i$th column
$A^\mathbf{k}_{ji},j=1,N+1$ with $\mathbf{k}={1,2,3}$ corresponds to a ``vector'' of eigenvectors
$x_i$, i.e. $\bx_i=\{x_{i+1},x_{i+2},x_{i+3}\}$ with eigenvalues
$\mu_{i+k}$ with $k=1,2,3$.
In the specific case of the
Gaussian chain $(\BA)_{ij}$ refer to
super-matrices
$(\BA)_{ij}\equiv A_{ij}\mathbbm{1}$, where the $i$th column
$A_{ji},j=1,M$ corresponds to an eigenvector of the
1-dimensional contraction of $\bxi$ (see e.g. Eq.~(\ref{matrix}) for
the Gaussian chain, i.e. $\bxi_{ii}\to 1$ and
$\bxi_{ii-1}\to(-1+(-1)^{\delta_{i,1}+\delta_{i,N+1}})$). 

In the particular case of the Gaussian chain with $M=3(N+1)$ the eigenvalues and
eigenvectors read
\begin{equation}
\mu_k=4\sin^2\left(\frac{k\pi}{2(N+1)}\right),\qquad
A_{ij}=\sqrt{\frac{2^{1-\delta_{j,0}}}{N+1}}\cos\left[\frac{(2i-1)j\pi}{2(N+1)}\right].
\end{equation}
 The back-transformation in general corresponds to
$(\br_i)_k=\sum_{j=1}^MA^k_{ij}(\bx_{j})_k$. In normal coordinates the the
 potential energy reads $U(\bxx)=\frac{1}{2}\sum_{k} \mu_k
 x_k^2$ while the
corresponding  Fokker-Planck equation for the evolution of the Green's
function of internal degrees of freedom (i.e. excluding center of mass
motion) at a temperature $T$, $G_T(\bx,t|\bx_0)$, reads
\begin{equation}
\left[\partial_t-D\sump
  \left(\partial^2_{\bx_k}+\beta\mu_k\partial_{\bx_k}\mathbf{x}_k\right)\right]G_T(\bx,t|\bx_0)=\delta(\bx-\bx_0),
\label{FPE}
\end{equation}
where $\beta=1/k_{\mathrm{B}}T$ and the primed sum runs over all
non-zero eigenvalues $\mu_k$, i.e. $\sump=\sum_{k; \mu_k\ne
  0}$.
Note that we are interested only in internal dynamics and not on the
center-of-mass dynamics, therefore we ignore in Eq.~(\ref{FPE}) and
what follows all contributions with $\mu_k=0$, as these pertain to
(ideal) rigid-body motions motion (i.e. center of mass translation
and rotation). Alternatively, we consider expansions around stable
minima, such that $\bxi$ is positive definite. Without any loss of generality we henceforth set $D=1$ and
measure energies in units of $k_{\mathrm{B}}T_\eeq$, where $T_\eeq$ is
the equilibrium (post-quench) temperature as defined in the
manuscript. Moreover, since we are only interested in the evolution at
temperature $T_\eeq$, we further express temperature relative to
$T_\eeq$, i.e. $\TT \equiv T/T_\eeq$, such that $\TT=1$ corresponds to
$T_\eeq$. 
The stationary solution of Eq.~(\ref{FPE}) corresponds to
the Boltzmann-Gibbs density
 \begin{equation}
 P^{\mathrm{eq}}_{\TT}(\mathbf{X})=\prodp\left(\frac{\mu_k}{2\pi}\right)^{3/2}\exp\left(-\frac{\mu_k\mathbf{x}_k^2}{2\TT}\right),
 \label{Gauss_eq}
 \end{equation}
 where $\prodp=\prod_{k; \mu_k\ne 0}$.
The probability density of $\bxx$ starting from an initial
probability density function $P^{\mathrm{eq}}_{\TT}(\mathbf{X})$ is
obtained from the Green's function via
\begin{equation}
  P_{\TT}(\bxx,t)=\!\int\!\! d\bxx_0
  G_1(\bxx,t|\bxx_0)P_{\TT}^{\eeq}(\bxx_0),
\label{GMarkov}  
\end{equation}
where 
\begin{equation}
   G_1(\mathbf{X},t|\mathbf{X}_0,0) = \prodp \left(\frac{\mu_k}{2\pi(1-\e{-2\mu_kt})}\right)^{3/2}\exp\left[-\frac{\mu_k}{2(1-\e{-2\mu_k t})}\left(\mathbf{x}_k^2-2\mathbf{x}_k\cdot\mathbf{x}_{0k}\e{-\mu_k t}+\mathbf{x}^{2}_{0k}\e{-2\mu_k t}\right)\right],
\end{equation}
is the well-known Green's function of an Ornstein-Uhlenbeck
process. Note that $\lim_{t\to\infty}G_{\TT}(\mathbf{X},t|\mathbf{X}_0,0)=P_{\TT}^{\eeq}(\bxx)$. The intergal \eql{GMarkov} can easily be performed
analytically and yields
\begin{equation}
  P_{\tilde{T}}(\mathbf{X},t)=\prodp\left(\frac{\mu_k}{2\pi[1+(\TT-1)\e{-2\mu_kt}]}\right)\exp\left(-\frac{\mu_k\mathbf{x}_k^2}{2[1+(\TT-1)\e{-2\mu_kt}]}\right).
  \label{GPDF}
  \end{equation}
\eql{GPDF} can now be used to calculate the Kullback-Leibler
divergence (Eq.~(3) in the Letter) to yield the first of Eqs.~(8) in
the Letter. Furthermore, the average potential energy and the system's
entropy are defined as
      \begin{equation}
       \langle U(t) \rangle_{\TT} \equiv\int d\mathbf{x}
       P_{\tilde{T}}(\mathbf{x},t) U(\mathbf{x}),\quad
       S_{\TT}(t)=-\int d\bx P_{\TT}(\bx,t)\ln P_{\TT}(\bx,t)
      \end{equation} 
      and read, upon performing the integration and introducing $\Lambda_k^{\TT}(t)\equiv
1+(\TT-1)\e{-2\mu_kt}$,
      \begin{equation}
       \langle U(t) \rangle_{\TT}=\frac{3}{2}\sump
       \Lambda_k^{\TT}(t), \quad S_{\TT}(t)=\frac{3}{2}\sump\left[1-\ln\left(\frac{\mu_k}{2\pi\Lambda_k^{\TT}(t)}\right)\right].
      \end{equation}       
      
In the projected, non-Markovian setting we are interested in the
dynamics of an internal distance $d_{ij}(t)\equiv
|\br_i(t)-\br_j(t)|$. In normal coordinates this corresponds to
\begin{equation}
d_{ij}\equiv |\br_i-\br_j|=\sump \left|\left(A_{ik}-A_{jk}\right) \bx_k\right|.
\label{distance}  
\end{equation}  
By doing so we project out $3(N-1)$ latent degrees of freedom and
track only $d_{ij}$. The 'non-Markovian Green's function', that is,
the probability density of $d_{ij}$ and time $t$ given that the full
system evolves from $P_{\TT}^{\eeq}(\bxx_0)$ is defined as
 \begin{eqnarray}
    \mathcal{P}_{\tilde{T}}(d,t)&=&\int d\Omega \int d\mathbf{X}_0\delta(\sump [A_{ik}- A_{jk}]\bx_k-\mathbf{d})
    G_1(\mathbf{X},t|\mathbf{X}_0,0)  P_{\TT}^{\eeq}(\mathbf{X}_0)
    \notag\\
    &=&d^2\int_0^{\infty}  dl_0 l_0^2\int d\Omega\int d\Omega_0\delta(\sump [A_{ik}-
      A_{jk}]\bx_k-\mathbf{d})\delta(\sump [A_{ik}-
      A_{jk}]\bx_{k,0}-\mathbf{l}_0)G_1(\mathbf{X},t|\mathbf{X}_0,0)
    P_{\TT}^{\eeq}(\mathbf{X}_0) \notag\\
    &\equiv&\int_0^{\infty} dl_0 \mathcal{P}_{\tilde{T}}(d,t,l_0;P_{\TT}^{\eeq}),
    \label{GausDis}
   \end{eqnarray}
 where we first project onto the vectors $\mathbf{d}$ and
 $\mathbf{d}_0$ and afterwards marginalize over all respective angles
 $\Omega$ and $\Omega_0$. Note that the stept in line 2 of
 \eql{GausDis} is actually not necessary but is preferable if one also wants to
 access the general non-Markovian  two-point joint density
 $\mathcal{P}_{\tilde{T}}(d,t,d_0;P_{\TT}^{\eeq})$. The calculation
 proceeds as follows.

 We first preform two 3-dimensional Fourier transforms $\mathbf{d}_0
 \to \mathbf{u}$ and $\mathbf{d}\to\mathbf{v}$:
 \begin{eqnarray}
   \label{FT}
   \hat{\mathcal{P}}_{\tilde{T}}(\mathbf{u},t,\mathbf{v}_0;P_{\TT}^{\eeq})&\equiv&\displaystyle{\frac{1}{(2\pi)^6}\int d\mathbf{d}\ee^{-i\mathbf{v}\cdot\mathbf{d}}\int d\mathbf{d}_0\ee^{-i\mathbf{u}\cdot\mathbf{d}_0}\mathcal{P}_{\tilde{T}}(\mathbf{d},t,\mathbf{d}_0;P_{\TT}^{\eeq})}\nonumber\\
   &=& \displaystyle{\frac{1}{(2\pi)^6}\prodp
       \exp\left[ -\frac{\mathcal{C}_k^{ij}}{2\mu_k}(1+(\TT-1)\e{-2\mu_k
           t})\mathbf{v}^2-\frac{\mathcal{C}^{ij}_k}{2\mu_k}\mathbf{u}^2-
        2\frac{\mathcal{C}^{ij}_k}{2\mu_k}\e{-\mu_kt}\mathbf{v}\cdot \mathbf{u}\right]\!,}
    \end{eqnarray}
 where we have introduced the short-hand notation
\begin{equation}
 \mathcal{C}_k^{ij}\equiv (A_{ik}-
 A_{jk})^2.
 \end{equation}
     Now we define, as in the main text, $\Lambda_k^{\TT}(t)\equiv
     1+(\TT-1)\e{-2\mu_kt}$ as well as 
\begin{equation}
\mathcal{A}_{\TT}^{ij}(t)\equiv\sump\Lambda_k^{\TT}(t)\mathcal{C}_k^{ij}/2\mu_k,\quad
\mathcal{B}_{\TT}^{ij}(t)\equiv\TT\sump \mathcal{C}_k^{ij}\e{-\mu_kt}/2\mu_k,
\end{equation}
and rewrite \eql{FT} as
\begin{equation}
 \hat{\mathcal{P}}_{\tilde{T}}(\mathbf{u},t,\mathbf{v}_0;P_{\TT}^{\eeq})=\frac{1}{(2\pi)^6}\exp\left(-\mathcal{A}_{\TT}^{ij}(t)\mathbf{v}^2-\mathcal{A}_{\TT}^{ij}(0)\mathbf{u}^2-2\mathcal{B}_{\TT}^{ij}(t)\mathbf{v}\cdot \mathbf{u}\right),
\label{FT2}  
\end{equation}
which can be easily inverted back to give
\begin{equation}
\mathcal{P}_{\tilde{T}}(\mathbf{d},t,\mathbf{d}_0;P_{\TT}^{\eeq})=(4\pi)^{-3}[\mathcal{A}_{\TT}^{ij}(t)\mathcal{A}_{\TT}^{ij}(0)-\mathcal{B}_{\TT}^{ij}(t)^2]^{-3/2}\exp\left(-\frac{1}{4}\frac{\mathcal{A}_{\TT}^{ij}(0)\mathbf{d}^2-2\mathcal{B}_{\TT}^{ij}(t)\mathbf{d}\cdot\mathbf{d}_0+\mathcal{A}_{\TT}^{ij}(t)\mathbf{d}_0^2}{\mathcal{A}_{\TT}^{ij}(t)\mathcal{A}_{\TT}^{ij}(0)-\mathcal{B}_{\TT}^{ij}(t)^2}\right).
\label{Gvec}  
\end{equation}
The marginalization is henceforth straightforward and yields
\begin{eqnarray}
  \mathcal{P}_{\tilde{T}}(d,t,d_0;P_{\TT}^{\eeq})&=&\frac{(dd_0)^2\exp\left(-\frac{1}{4}\frac{\mathcal{A}_{\TT}^{ij}(0)\mathbf{d}^2+\mathcal{A}_{\TT}^{ij}(t)\mathbf{d}_0^2}{\mathcal{A}_{\TT}^{ij}(t)\mathcal{A}_{\TT}^{ij}(0)-\mathcal{B}_{\TT}^{ij}(t)^2}\right)}{2\pi[\mathcal{A}_{\TT}^{ij}(t)\mathcal{A}_{\TT}^{ij}(0)-\mathcal{B}_{\TT}^{ij}(t)^2]^{3/2}}\int_0^{\pi}d\cos\theta\exp\left(\frac{1}{2}\frac{dd_0\mathcal{B}_{\TT}^{ij}(t)\cos\theta}{\mathcal{A}_{\TT}^{ij}(t)\mathcal{A}_{\TT}^{ij}(0)-\mathcal{B}_{\TT}^{ij}(t)^2}\right)\nonumber\\
  &=&\frac{dd_0}{2\pi\mathcal{B}_{\TT}^{ij}(t)}\frac{\exp\left(-\frac{1}{4}\frac{\mathcal{A}_{\TT}^{ij}(0)\mathbf{d}^2+\mathcal{A}_{\TT}^{ij}(t)\mathbf{d}_0^2}{\mathcal{A}_{\TT}^{ij}(t)\mathcal{A}_{\TT}^{ij}(0)-\mathcal{B}_{\TT}^{ij}(t)^2}\right)}{[\mathcal{A}_{\TT}^{ij}(t)\mathcal{A}_{\TT}^{ij}(0)-\mathcal{B}_{\TT}^{ij}(t)^2]^{1/2}}\sinh\left(\frac{1}{2}\frac{\mathcal{B}_{\TT}^{ij}(t)dd_0}{\mathcal{A}_{\TT}^{ij}(t)\mathcal{A}_{\TT}^{ij}(0)-\mathcal{B}_{\TT}^{ij}(t)^2}\right).
\label{Gvec}  
\end{eqnarray}
The probability density of $d$ at time $t$ after having started from
an initial density $ P^{\mathrm{eq}}_{\TT}(\mathbf{X}_0)$ (i.e. the
pre-quench equilibrium) follows by simple integration and finally reads
\begin{equation}
\mathcal{P}_{\tilde{T}}(d,t)=\int_0^{\infty} dl_0 \mathcal{P}_{\tilde{T}}(d,t,l_0;P_{\TT}^{\eeq})\equiv\frac{d^2}{2\sqrt{\pi}}\mathcal{A}_{\TT}^{ij}(t)^{-3/2}\mathrm{e}^{-d^2/4\mathcal{A}_{\TT}^{ij}(t)},
\end{equation}  
which is precisely Eq.~(7) in the manuscript. The average potential of mean
force, $\langle \mU(t)\rangle_{\TT}\equiv-\langle \ln
\mathcal{P}^{\eeq}_1(d)\rangle_{\TT}$ and entropy, $\mS_{\TT}(t)\equiv
-\langle \ln \mpp_{\TT}(d,t) \rangle_{\TT}$ (in
units of $k_{\mathrm{B}}T$), where $\langle f(d)\rangle_{\TT}\equiv
\int dl \mpp_{\TT}(l,t)f(l) $, in turn read
\begin{eqnarray}
\langle
\mU(t)\rangle_{\TT}&=&\ln\left(2\sqrt{\pi}\mathcal{A}_{\TT}^{ij}(0)^{3/2}\right)-\mathcal{A}_{\TT}^{ij}(t)^{1/2}(2-\gamma_\mathrm{e}+\ln\mathcal{A}_{\TT}^{ij}(t)
  )+\frac{3}{2}\frac{\mathcal{A}_{\TT}^{ij}(t)}{\mathcal{A}_{\TT}^{ij}(0)}\nonumber\\
\mS_{\TT}(t)&=&\ln\left(2\sqrt{\pi}\mathcal{A}_{\TT}^{ij}(t)^{3/2}\right)-\mathcal{A}_{\TT}^{ij}(t)^{1/2}(2-\gamma_\mathrm{e}+\ln\mathcal{A}_{\TT}^{ij}(t)
  )+\frac{3}{2}
\label{PMF}  
\end{eqnarray}  
where $\gamma_\mathrm{e}$ denotes Euler's gamma.
     Using the results in \eql{PMF} as well as the definition of the
     equilibrium free energy, $F=-\ln Q_{1}\equiv -\ln\int d\bxx
     \mathrm{e}^{-U(\bxx)}$, (where all potentials are in units of $k_\mathrm{B}T_{\eeq}$) we arrive at
     \begin{equation}
       \mathcal{D}[P_{\tilde{T}}(t)||P_1]=\langle U_{\tilde{T}}(t)\rangle-S_{\tilde{T}}(t) - F,\quad\mathcal{D}[\mathcal{P}_{\tilde{T}}(t)||\mathcal{P}_1]=\langle \mathcal{U}_{\tilde{T}}^{\mathrm{eff}}(t)\rangle-\mathcal{S}_{\tilde{T}}(t),
     \end{equation}
     which are exactly Eqs.~(4) and (5) in the Letter. For any stable
     symmetric matrix $\bxi$ the condition of equidistant quenches
     $\mathcal{D}[P_{\TT^+}(0^+)||P_1]=\mathcal{D}[P_{\TT^-}(0^+)||P_1]$
     is satisfied by
\begin{equation}
     \TT^+-\TT^-=\ln(\TT^+/\TT^-) \quad\to\quad
     \TT^+(\TT^-)=-W_{-1}(-\TT^-\ee^{-\TT^-}),
\end{equation}     
     where $W_{-1}(x)$
     defined for $x\in[-\ee^{-1},0)$ denotes the second real branch of
     the Lambert-W function, which in turn satisfies the following sharp two-sided
     bound \cite{SWbound}
 \begin{equation}    
\frac{2}{3}\left[1 + \sqrt{2(\TT^- - 1 - \ln\TT^-)} + \TT^- - 1 -\ln\TT^-\right]\le
 \TT^+(\TT^-)\le 1 + \sqrt{2(\TT^- - 1 - \ln\TT^-)} + \TT^- - 1 -\ln\TT^-.
   \label{bound}
 \end{equation}
 
\subsection{Kullback-Leibler divergence and uphill/downhill asymmetry in relaxation of a random Gaussian network}
In the Letter we prove that for any reversible ergodic
Ornstein-Uhlenbeck process uphill relaxation (i.e. for a quench from
$\TT^-\uparrow 1$ for which $\langle U(0^+)\rangle_{\TT^-}-\langle
U\rangle_1<0$) is always faster that downhill relaxation (i.e. for a quench from
$\TT^+\downarrow 1$ for which $\langle U(0^+)\rangle_{\TT^-}-\langle
U\rangle_1>0$), where the pair of equidistant quenches $\TT^+$ and
$\TT^-$ is defined in the Letter. To visualize this on hand of an
additional instructive example, we generated a random Gaussian network
with 10 beads by filling elements of the upper-triangular part of the connectivity
matrix with a $-1$ according to a Bernulli distribution with
$p=0.7$. The resulting matrix was then symmetrized and the diagonal
elements chosen to assure sure mechanical stability
(i.e. 'connectedness'). The resulting connectivity matrix
$\boldsymbol{\Gamma}$ is related to the general Ornstein-Uhlenbeck
matrix in \eql{OUP} via $\bxi=\boldsymbol{\Gamma}\otimes\mathbbm{1}$, where
\begin{equation}\boldsymbol{\Gamma}=
  \begin{pmatrix}
  5 & -1 & -1 & 0 & -1 & 0 & 0 & 0 & -1 & -1 \\
  -1 & 5 & -1 & 0 & -1 & -1 & 0 & 0 & 0 & -1 \\
  -1 & -1 & 8 & -1 & -1 & 0 & -1 & -1 & -1 & -1 \\
  0 & 0 & -1 & 7 & -1 & -1 & -1 & -1 & -1 & -1 \\
  -1 & -1 & -1 & -1 & 9 & -1 & -1 & -1 & -1 & -1 \\
  0 & -1 & 0 & -1 & -1 & 7 & -1 & -1 & -1 & -1 \\
  0 & 0 & -1 & -1 & -1 & -1 & 7 & -1 & -1 & -1 \\
  0 & 0 & -1 & -1 & -1 & -1 & -1 & 7 & -1 & -1 \\
  -1 & 0 & -1 & -1 & -1 & -1 & -1 & -1 & 8 & -1 \\
  -1 & -1 & -1 & -1 & -1 & -1 & -1 & -1 & -1 & 9
  \end{pmatrix}.
     \end{equation}
  
The corresponding results for
$\mathcal{D}[\mathcal{P}_{\tilde{T}}(t)||\mathcal{P}_1]$, whereby we
tagged the distance between the 1st and 10th bead, i.e. $d=|\br_1-\br_{10}|$ are shown in Fig. \ref{random}.
     \begin{figure}[!ht]
     \begin{center}
      \includegraphics[width=0.4\textwidth]{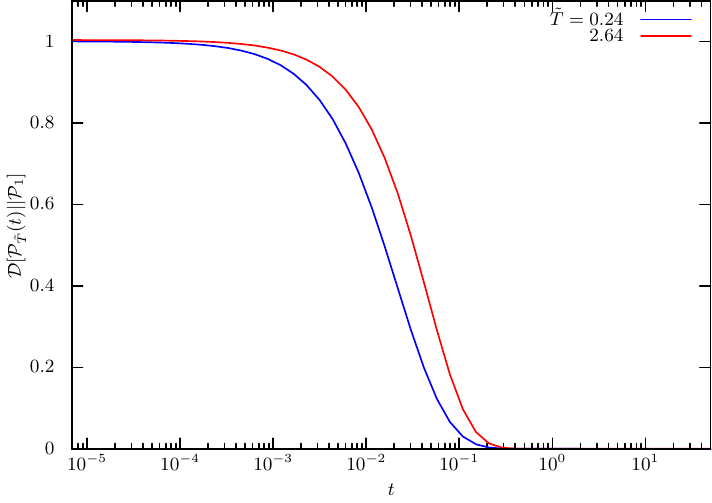}
      \caption{$\mathcal{D}[\mathcal{P}_{\TT^{\pm}}(t)||\mathcal{P}_1]$
        as a function of time for a pair of equidistant quenches
        with $\TT^+=2.64$ and $\TT^-=0.24$, which illustrates the asymmetry in the thermal
        relaxation holds for any Gaussian Network (according to our proof).}
      \label{random}
     \end{center}
     \end{figure}

     \section{Tilted Single File}
      We consider a system of $N$ hard-core point-particles (the extension to a finite
  diameter is straightforward
  \cite{Slizana_diffusion_2009,Slizana_single-file_2008}) diffusing in a
  box of unit length with a diffusion coefficient $D$, which we set
  equal to 1 and express energies in units of $k_{\mathrm{B}}T_{\eeq}$ without any loss of generality. The particles with positions $\bx=\{x_i\}$ feel the presence of a linear
  potential $U(\{x_i\})=\sum_{i=1}^N g x_i$. The Green's function of
  the system obeys the many-body Fokker-Planck equation
  \begin{equation}
    (\partial_t-\hat{\mathcal{L}}_{\TT})G_{\TT}(\bx,t|\bx_0)\equiv\left(\partial_t-\sum_{i=1}^N(\partial_{x_i}^2+g\TT^{-1}\partial_{x_i})\right)G_{\TT}(\bx,t|\bx_0)=\delta(\bx-\bx_0) 
    \label{FPEs}
    \end{equation}
  The confining walls are
  assumed to be perfectly reflecting, i.e
  $J(x_i)|_{x_i=0}=J(x_i)|_{x_i=1}=-D(g/\TT-\partial_{x_i})G_{\TT}(\bx,t|\bx_0)=0,\forall
  i$. Moreover, particles are not allowed to cross, which introduces
  the following set of internal boundary conditions
 \begin{equation}
\left(\partial_{x_{i+1}}-\partial_{x_i}\right)G_{\TT}(\bx,t|\bx_0)|_{x_{i+1}-x_i=0}=0,\forall
i.
\label{IBC}
 \end{equation}
\eql{FPEs} with reflecting external boundary conditions
$J(x_i)|_{x_i=0}=J(x_i)|_{x_i=1}=0,\forall i$ and internal boundary
conditions in \eql{IBC} is solved exactly using the coordinate Bethe
ansatz (we do not repeat the results here as they can be found in
\cite{Slapolla_manifestations_2019}). It is convenient to introduce the
particle-ordering operator
\begin{equation}
\hat{\mathcal{O}}_{\bx}\equiv\prod_{i=2}^N\theta(x_i-x_{i-1}),  
\label{order}
\end{equation}
where $\theta(x)$ is the Heaviside step-function. Let
$\zeta_{\TT}(x_i,t|x_{0i})$ be the Green's function of the corresponding
single-particle problem and $P^\eeq_{\TT}(x_i)=\lim_{t\to\infty}\zeta_{\TT}(x_i,t|x_{0i})$ the density of the
equilibrium measure at temperature $\TT$, then the Green's
function can be written directly as
\begin{equation}
G_{1}(\bx,t|\bx_0)=N!\hat{O}_{\bx}\prod_{i=1}^N\zeta_1(x_i,t|x_{0i})\to P_{\TT}(\bx,t)=N!\hat{O}_{\bx}\prod_{i=1}^N\int_0^1dx_{i0}\zeta_1(x_i,t|x_{i0})P^\eeq_{\TT}(x_{i0}),
\end{equation}
where the normalization factor $N!$ assures a correct re-weighing of
non-crossing trajectories \cite{Slapolla_manifestations_2019}. We expand the
  Green's function for a single particle at $\TT=1$ in a bi-orthonormal
  eigenbasis, $\zeta(x,t|x_0)=\sum_k \phi_k^R(x)
  \phi_k^L(x_0)\ee^{-\lambda_kt}$, where $\lambda_0=0$, $
  \lambda_k=\pi^2k^2+g^2/4$ and 
\begin{equation}
  \phi^L_k(x)=\frac{\ee^{gx/2}}{\sqrt{2\lambda_k}}
  \left(g\sin(k\pi x)-2k\pi\cos(k\pi x)\right),  k>0
 \label{eigen} 
   \end{equation}
and $\phi^R_k(x)=\e{-gx}\phi^L_k(x)$, whereas for $k=0$ we have
$\phi^L_0(x)=1$, $\phi^R_0(x)=P_1^{\mathrm{eq}}(x)$.

A key simplification in the calculation of order-preserving integrals as well as all projected, tagged-particle
observables (incl. functionals; see
e.g. \cite{Slapolla_manifestations_2019}) is the so-called 'extended
phase space integration' introduced by Lizana and Ambj\"ornsson
\cite{Slizana_diffusion_2009,Slizana_single-file_2008}, according to
which for any $1\le M\le N$ and some function $f(\bx)$ that is
symmetric with respect to permutation of coordinates $x_i$
\begin{equation}
\hat{O}_{\bx}\prod_{i=1}^N\int_0^1dx_{i0}f(\bx)\delta(z-x_M)=\prod_{i=1}^{M-1}\int_0^zdx_{i0}\prod_{j=M+1}^{N}\int_z^1dx_{j0}\frac{f(x_M=z,\{x_{i\ne
    M}\})}{(M-1)!(N-M)!}.
\label{extended}  
\end{equation}  
With the aid of \eql{extended} it is possible to calculate the
Kullback-Leibler divergence as
\begin{equation}
\mathcal{D}[P_{\tilde{T}}||P_1]=\int d\bx P_{\tilde{T}}(\mathbf{x},t)
\ln(P_{\tilde{T}}(\bx,t)/P_1(\mathbf{x}))\equiv
\left[\int_0^1dxP^1_{\TT}(x,t)\ln(P^1_{\TT}(x,t)/P^1_{1}(x))\right]^N,
\label{KLS}
\end{equation}  
where
$P^1_{\TT}(x,t)=\int_0^1dx_{0}\zeta_1(x,t|x_{0})P^\eeq_{\TT}(x_0)$,
and the second equality is a result of applying \eql{extended}. The
result in \eql{KLS} for a single file of 10 particles is depicted in
Fig.~(3a) in the Letter. For the sake of completeness, we also present
the exact explicit result for $\langle U(t)\rangle_{\TT}\equiv
gN\langle x(t)\rangle_{\TT}$, which reads
\begin{equation}
\langle
U(t)\rangle_{\TT}=gN\left(\frac{1-\mathrm{e}^{g}+g}{g(1-\mathrm{e}^g)}+8\sum_{k=1}^{\infty}\left(\frac{gk\pi}{\lambda_k}\right)^2\frac{(\TT-1)(\ee^{g/2}-(-1)^k)(\ee^{g/\TT}-(-1)^k\ee^{g/2})}{(\ee^{g/\TT}-1)[(\TT-2)^2g^2+(2\pi
    k\TT)^2]}\ee^{-\lambda_kt}\right).
\label{meanU}
\end{equation}
%
     
The results for the non-Markovian tagged-particle dynamics can be
derived analogously. The probability density function for tagging the
$M$th particle  is defined as
\begin{equation}
\mpp_{\TT}(z,t)=\hat{\Pi}_{\mathbf{x}}(z)P_{\tilde{T}}(\bx,t)\equiv\hat{O}_{\bx}\prod_{i=1}^N\int_0^1dx_{i0}\delta(z-x_{\mathcal{T}})P_{\tilde{T}}(\bx,t)
  \label{NMS}
\end{equation}
and since $P_{\tilde{T}}(\bx,t)$ is symmetric to permutation of
particle indices \eql{extended} can be applied. The exact result has
the form of a spectral expansion and reads
\begin{equation}
\mathcal{P}_{\TT}(z,t)=\sum_{\mathbf{k}}
V_{0\mathbf{k}}(z)\mathcal{V}_{\mathbf{k}0}^{\TT}
\mathrm{e}^{-\lambda_\mathbf{k} t},
\label{bethe}
\end{equation}
where  $\mathbf{k}=\{k_i\}$ is a $N$-tuple of non-negative integers
and $\lambda_\mathbf{k}=\sum_{n=1}^N \lambda_{k_n}$ are Bethe eigenvalues of the
  operator $\hat{\mathcal{L}_1}=\sum_{i=1}^N(\partial_{x_i}^2+g\partial_{x_i})$ in a unit
  box under non-crossing conditions with $\lambda_0=0$ and $
  \lambda_{k_i}=\pi^2k_i^2+g^2/4,\forall k>0$. Let $N_L=\mathcal{T}-1$ and $N_R=N-\mathcal{T}$ be the total number of particles to the left
  and to the right of the tagged particle, respectively. Then $V_{0\mathbf{k}}(z)$ and
  $\mathcal{V}_{\mathbf{k}0}^{\TT}$ in \eql{bethe} are defined as 
\begin{align}    
 \label{coeff1} 
 &\!\!\!V_{0k}(z)=\frac{m_{\mathbf{k}}}{N_L!N_R!}\frac{2g\alpha}{\TT\Omega^g_{\TT}(0,1)}\sum_{\{k_i\}}
T^{1}_{\mathcal{T}}(z)\prod_{i=1}^{N_L}L_i^{1}(z)\!\!\!\!\!\!\prod_{i=N_L+2}^N
\!\!\!\!\!R_i^{1}(z)\\
&\!\!\!\mathcal{V}_{\mathbf{k}0}^{\TT}=\frac{N!}{N_L!N_R!}\frac{2g\alpha}{\TT\Omega^g_{\TT}(0,1)}\!\int_0^1\!\!dz\sum_{\{k_i\}}
T^{\TT}_{\mathcal{T}}(z)
\prod_{i=1}^{N_L}L_i^{\TT}(z)\!\!\!\!\!\!\prod_{i=N_L+2}^N
\!\!\!\!\!R_i^{\TT}(z)
\label{coeff2}
\end{align}
where $\alpha=\TT/(2-\TT)$,  $\Omega^g_{\TT}(x,y)\equiv
\ee^{-gx/\TT}-\ee^{-gy/\TT}$, and $m_\mathbf{k}=\prod_i n_{k_i}!$ is the  multiplicity
  of the Bethe eigenstate corresponding to the $N$-tuple $\mathbf{k}$,
  and the number $n_{k_i}$ counts how many times the eigenindex $k_i$
  appears in the Bethe eigenstate
  \cite{Slapolla_manifestations_2019}. In \eql{coeff2} we have introduced the auxiliary
functions
\begin{equation}
T^{\TT}_{\mathcal{T}}(z)=P_{\TT}^{\mathrm{eq}}(z)\frac{\ee^{gx/2}\left(g\sin(k_{\mathcal{T}}\pi
z)-2k_{\mathcal{T}}\pi\cos(k_{\mathcal{T}}\pi
z)\right)}{\sqrt{2\lambda_{k_{\mathcal{T}}}}},\forall k_{\mathcal{T}}>0
\end{equation}
and $T^{\TT}_{\mathcal{T}}(z)=P_{\TT}^{\mathrm{eq}}(z)$ for $k_{\mathcal{T}}>0$ 
where $P_{\TT}^{\mathrm{eq}}(z)$ is defined as
\begin{equation}
\mathcal{P}_{\tilde{T}}^{\mathrm{eq}}(z)=\frac{g N!}{N_L!N_R!}\frac{\Omega^g_{\TT}(0,z)^{N_L}\Omega^g_{\TT}(z,1)^{N_R}}{\TT\Omega^g_{\TT}(0,1)}\ee^{-gz/\TT},
\label{ragged}  
\end{equation}   
as well as
\begin{eqnarray}
    L_i^{\TT}(z)&=&\begin{cases}
    \Omega_{\TT}^g(0,z)/\Omega_{\TT}^g(0,1), \quad k_i=0\\
\lambda_{\sqrt{\alpha}k_i}\Phi_{k_i}^{g,\alpha}(0,z)+k_i\pi g\Psi_{k_i}^{g,\alpha}(0,z)(\TT-1)/(2-\TT),  \quad k_i>0\notag\\    
    \end{cases}\\
    R_i^{\TT}(z)&=&\begin{cases}
    \Omega_{\TT}^g(z,1)/\Omega_{\TT}^g(0,1), \quad k_i=0\notag\\\
    \lambda_{\sqrt{\alpha}k}\Phi_{k_i}^{g,\alpha}(z,1)+k_i\pi g\Psi_{k_i}^{g,\alpha}(z,1)(\TT-1)/(2-\TT),\quad k_i>0,\notag\
    \end{cases}
  \end{eqnarray}  Note that
  $\lambda_{xk}\equiv\pi^2(x k_i)^2+g^2/4,\forall k>0$, and $\sum_{\{k_i\}}$ denotes the sum over all possible
  permutations of $\mathbf{k}$ and the functions
  $\Phi_k^{g,\alpha}(x,y)$ and $\Psi_k^{g,\alpha}(x,y)$ are defined as
\begin{eqnarray}
\!\!\!\!\!\!\!\Phi_k^{g,\alpha}(x,y)&=&\frac{\ee^{-gx/2\alpha}\sin(k\pi
x)-\ee^{-gy/2\alpha}\sin(k\pi y)}{\lambda_{\alpha k}\sqrt{2\lambda_{k}}}\nonumber\\
\!\!\!\!\!\!\!\Psi_k^{g,\alpha}(x,y)&=&\frac{\ee^{-gx/2\alpha}\cos(k\pi
x)-\ee^{-gy/2\alpha}\cos(k\pi y)}{\lambda_{\alpha k}\sqrt{2\lambda_{k}}}.
\label{auxx}
\end{eqnarray}  
Details of the calculations can be found in
\cite{Slapolla_manifestations_2019}. The evaluation of Kullback-Leibler
divergence, $S_{\TT}(t),\mS_{\TT}(t)$ as well as
$\langle\mU(t)\rangle_{\TT}$ cannot be carried out analytically and we
therefore resort to efficient and accurate numerical quadratures. The
results are presented in Fig.~(3) in the Letter.

We performed extensive systematic calculations for different values of $g$ and $N$,
various combinations of $\TT^{\pm}$ as well as for different choices
for tagged particles. All these calculations gave the same qualitative
picture -- \emph{without any exceptions 'uphill' relaxation was always
  faster.} However, we are not able to prove rigorously that this is indeed
always the case. Therefore, for the single file the universally faster
uphill relaxation is only a conjecture.

\section{Non-existence of a unique relaxation asymmetry in multi-well
  potentials and generic conditions when the asymmetry is obeyed} 

In the letter we demonstrated that the relaxation in single-well
potentials is faster uphill than downhill. We have proven that this is
always the case near stable minima and for any reversible
Ornstein-Uhlenbeck process. Based on additional physical arguments we
hypothesized that the asymmetry is a general feature of diffusion in single-well
potentials. However, as we remarked in the Letter, it is not difficult to construct
counterexamples proving that the asymmetry is \emph{not} a general phenomenon in
all reversible ergodic diffusion processes.  

To that end we condider Markovian diffusion in rugged, multi-well
potentials parametrized by
     \begin{equation}
       U(x)=e(ax^6+bx^4+cx^3+dx^2),
       \label{multi}
     \end{equation}
     with some appropriately chosen constants $a,b,c,d$ and $e$. Let
     the dynamics evolve according to
     $\hat{L}_{\TT}=\partial^2_x-\TT^{-1}\partial_xF(x)$,
     where $F(x)=-6e(ax^5+4bx^3+3cx^2+2dx)$
     in a finite domain $a\le x\le b$ with reflecting boundaries,
     and let the corresponding Green's function be the
     solution of the following initial-boundary value problem
\begin{equation}
(\partial_t- \hat{L}_{\TT})G_{\TT}(x,t|x_0)=\delta(x-x_0), \quad
  -(\partial_x-\TT^{-1}F(x))G_{\TT}(x,t|x_0)|_{x=a}=-(\partial_x-\TT^{-1}F(x))G_{\TT}(x,t|x_0)|_{x=b}=0.
  \label{SFPE}
\end{equation}
We solve the Fokker-Plank equation so defined via the Method of Lines. The results for three distinct
parameter sets is shown in Fig.~(\ref{multiwell}).
     \begin{figure}[!ht]
      \begin{center}
       \includegraphics[width=0.85\textwidth]{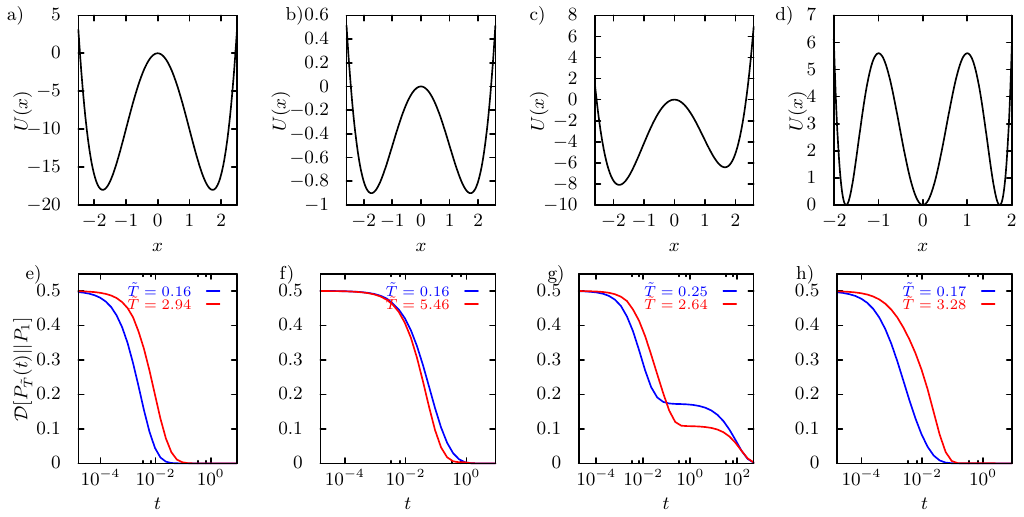}
      \end{center}
      \caption{In panels a,b) and e,f) the potential is a quartic with
        parameter $a=0,
        b=1, c=0, d=-6, e=2$ in panels a and f and$a=0,
        b=1, c=0, d=-6, e=0.1$ in panels b and f. In the asymmetric potential in panels c and
        g with $a=0,b=1,c=0.2,d=-6,e=0.8$ and panels c and f with
        $a=1,b=-6,c=0,d=9,e=1.4$, respectively,  the single-well
        asymmetry-pattern in fact becomes reversed. In a tripple-well
        with equally deep wells the asymmetry is again obeyed despite
        the middle well being wider.}
      \label{multiwell}
     \end{figure}
      We did not perform a systematic analysis of all the possible
      potentials. However, based on our observations it seems that the
      different uphill/downhill relaxation patterns depend on how different
      entropic contributions (i.e. intra-well entropy versus
      inter-well configuration entropy) change qualitatively with temperature for
      potentials with several minima.

      If we focus on the asymmetric case (Fig.~\ref{multiwell}c) we
      find that uphill relaxation is initially always faster, which is
      a direct result of the physical mechanism at play that we present in
      the Letter. At longer time the asymmetry gets inverted by the
      slow inter-well partitioning of probability mass. It is now not
      diffcult to understand that by making the asymmetry smaller we
      will move the crossing point, where the corves intersect, closer
      to $\mathcal{D}[P_{\TT^{\pm}}(t)||P_1^{\eeq}]=0$, such that for a
      sufficiently small asymmetry -- which in the letter we refer to
      \emph{near degeneracy} -- uphill relaxation will eventually be
      faster for all times, for which
      $\mathcal{D}[P_{\TT^{\pm}}(t)||P_1^{\eeq}]$ differs from zero by
      an amount that is not neglgible/detectable. For a formal
      discussion of this situation see below.

It is interesting and important to note that the asymmetry is also
obeyed if the barrier is moderately high, i.e. such that a small but non-neglible
probability mass is located at the barrier (see
Fig.~\ref{1D_double}). However, the quench must then not be too
strong. That is, an 'infinitely' high barrier effecting a strict time-scale separation
betwenn intra-well and inter-well relaxation is \emph{not} a
neccessarry condition for the asymmetry to occur. To demonstrate this
we inspect overdamped relaxation according to Eq.~(\ref{SFPE}) in the
following double well potential
$U(x)=\Delta(x^2-1)^2$, where we choose (in units of $k_{\rm B}T_{\rm
  eq}$) $\Delta=3$ and $F(x)=-U(x)'$.

     \begin{figure}[!ht]
      \begin{center}
       \includegraphics[width=0.5\textwidth]{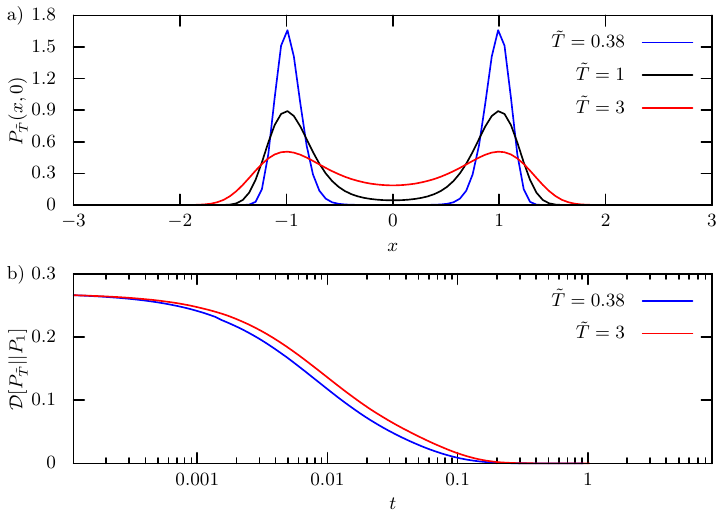}
      \end{center}
      \caption{a) Density of invariant measure at $\tilde{T}=1$
        (i.e. equilibrium probability density), and the equidistant post-quench
        probability densities at $\tilde{T}^+=3$ and
        $\tilde{T}^-=0.38$; b) Corresponding time evolution of the
        Kullback-Leibler divergence depicting that the asymmetry is obeyed.}
      \label{1D_double}
     \end{figure}

In order to check that the observed effect in multi-well potentials
is not an artifact of one-dimensional systems now also inspect
2-dimensional multi-well potentials. To that end we consider 4-well
potentials parametrized by
\begin{equation}
  U(x,y)=\Delta_x(x^2-x_0^2)^2+\Delta_y(y^2-y_0^2),
\label{4well}  
\end{equation}
where energy is measured in units of $k_{\rm B}T_{\rm
  eq}$. We solve the problem by the \emph{Alternating Direction
  Implicit} method (ADI) developed in \cite{SAG_ADI} with 4-step
operator splitting. 
We first focus on the limit of high barriers and quenches
leaving the inter-well partitioning of probability mass unaffected
(see Fig.~\ref{2D_double1}). According to the proposed principle and
prediction the symmetry is obeyes and uphill relaxation is always
faster. 

     \begin{figure}[!ht]
      \begin{center}
       \includegraphics[width=0.85\textwidth]{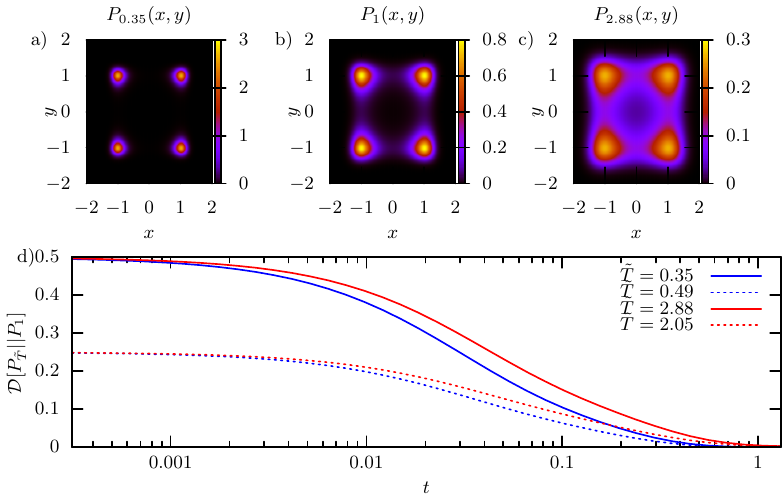}
      \end{center}
      \caption{Density of invariant measure at $\tilde{T}=1$ (b)
        (i.e. equilibrium probability density), and the equidistant post-quench
        probability densities at (c) $\tilde{T}^+=2.88$ and
        (a) $\tilde{T}^-=0.35$ for the 4-well potential in
        Eq.~(\ref{4well}) with parameters $\Delta_x=\Delta_y=3$ and $x_0=y_0=1$.; d) Corresponding time evolution of the
        Kullback-Leibler divergence depicting that the asymmetry is
        obeyed for two pairs of equidistant temperatures.}
      \label{2D_double1}
     \end{figure}

In Fig.~\ref{2D_mpemba} now inspect the case of a moderately high barriers (where the
probability density on top of the barriers does not vanishes). As
expected the asymmetry is obeyed only for sufficiently small quenches,
whereas it becomes violated for stronger quenches (compare full and
dashed lines). The reason for the violation is the fact that the
inter-well redistribution becomes the dominant step for strong
quenches.
     
     \begin{figure}[!ht]
      \begin{center}
       \includegraphics[width=0.85\textwidth]{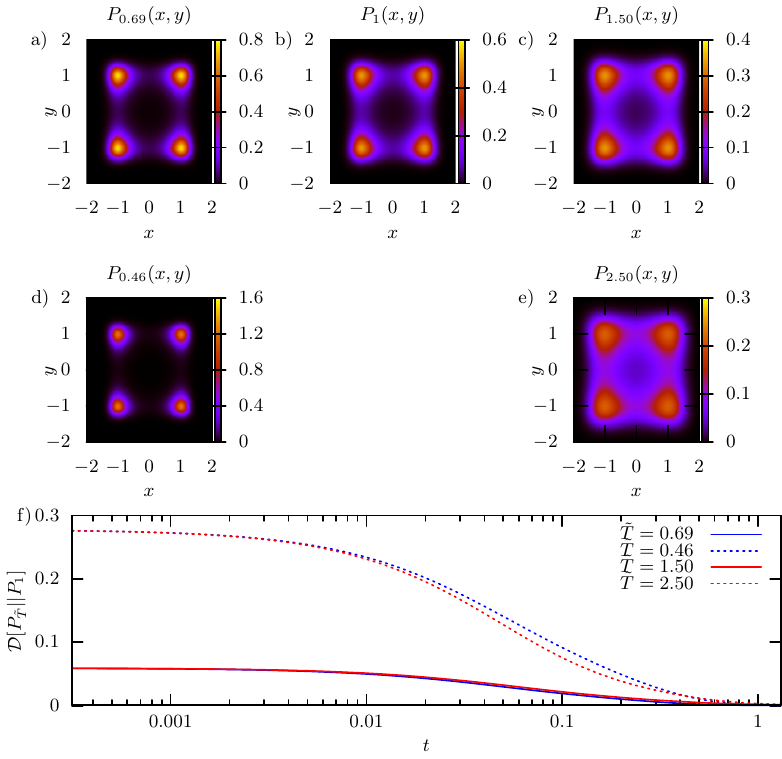}
      \end{center}
      \caption{Density of invariant measure at (a) $\tilde{T}^-=0.69$
        (b) at $\TT=1$, and at (c) $\tilde{T}^+=1.5$, (d)
        $\tilde{T}^-=0.46$ and (e) $\tilde{T}^+=2.5$ corresponding to
        the 4-well potential in 
        Eq.~(\ref{4well}) with parameters $\Delta_x=\Delta_y=2$ and $x_0^2=y_0^2=1$.; b) Corresponding time evolution of the
        Kullback-Leibler divergence depicting that the asymmetry is
        obeyed for small quenches (a and c) and violated for strong
        quenches (d and e).}
      \label{2D_mpemba}
     \end{figure}  

It seems that the asymmetry
      observed in single-well potentials persists in nearly degenerate
      potentials and ceases to exists as soon as the potential
      becomes sufficiently asymmetric with sufficiently deep wells, where entropy
      attains an additional inter-well configurational component, such
      that during relaxation the probability mass becomes
      re-distributed between the wells in an asymmetric manner. 

\subsection{The asymmetry is obeyed in degenerate potentials in the presence of a time-scale separation}

We now provide also formal arguments confirming that the symmetry must
be obeyed in degenerate potentials in the presence of a time-scale
separation. We follow the work of Moro \cite{SMoro}. Since we are
dealing with systems obeying detailed balance the generator of the
relaxation dynamics $\LL$ is always diagonalizable, i.e.
\begin{equation}
\LL_{T}=\sum_{k\ge 0}-\lambda_k\psi_k^{R}(\bx)\psi_k^{L}(\bx_0)
\end{equation}
where $\psi_k^{R}(\bx)$ and $\psi_k^{L}(\bx)$ are the orthonormal right and left
eigenfunctions, respectively, (i.e. $\int
\psi_k^{L}(\bx)\psi_l^{R}(\bx)d\bx=\delta_{kl}$) and $-\lambda_k$ are
real eigenvalues ($\lambda_0=0$ as we have assumed that the potential
is confining and the dynamics is ergodic). The eigenfunctions constitute a complete bi-orthonormal
basis, $\sum_k\psi_k^{L}(\bx)\psi_k^{R}(\bx')=\delta(\bx-\bx')$. As a
result of detailed balance we have
$\psi_k^{R}(\bx)=\e{-U(\bx)/k_{\rm B}T}\psi_k^{L}(\bx)$ and
$\psi_0^{R}(\bx)=P_{T}^{\eeq}\equiv\e{-U(\bx)/k_{\rm B}T}/\int\e{-U(\bx)/k_{\rm B}T}d\bx$ and $\psi_0^{L}(\bx)=1$. Let
$\LL_T^{\dagger}$ be the adjoint (or 'backward') generator, then we have
the pair of eigenproblems
$\LL_T\psi_k^{R}(\bx)=-\lambda_k\psi_k^{R}(\bx)$ and
$\LL^{\dagger}_T\psi_k^{L}(\bx)=-\lambda_k\psi_k^{L}(\bx)$.

The Green's function of the relaxation problem,
$(\partial_t-\LL_T)G_T(\bx,t|\bx_0)=0$ with $G_T(\bx,0|\bx_0)=\delta(\bx-\bx_0)$, decomposes to
\begin{equation}
G_T(\bx,t|\bx_0)=\sum_{k\ge
  0}\psi_k^{R}(\bx)\psi_k^{L}(\bx_0)\e{-\lambda_kt}\to \quad
P_{\TT}(\bx,t)=\int G_1(\bx,t|\bx_0)P_{\TT}^{\eeq}(\bx_0)d\bx_0.
\end{equation}
In presence of a time-scale separation (as a result of the existence
of one or more high energy barriers) the eigenvalue spectrum of $\LL$
has a gap, i.e. $\exists k_{\rm min}$ such that $\lambda_{k_{\rm
    min}+l}\gg k_{\rm min}\forall l\ge 1$. 

Assume now a set of $M$ well-defined deep minima at $\hat{\bx}_{i},
i=1,\ldots,M$. This implies $k_{\rm min}=M-1$. Let us define \emph{localizing functions} $g_i(\bx),i\in[1,M]$
such that
\begin{equation}
c_i^{\eeq}\equiv\int g_i(\bx)P_{1}^{\eeq}(\bx) d\bx\quad \to\quad \int
g_i(\bx)\psi_k^{R}(\bx)=0,\forall k\ge M, 
  \label{cs}
\end{equation}
$c_i^{\eeq}$ are the equilibrium site populations.
The localizing functions therefore by definition separate the
intra-well relaxation from the inter-well 'hopping' of probability
mass. In turn this implies that $g_i(\bx)$ belong to the subspace
$\{\psi_k^{L}(\bx)\},k< M$, i.e.
\begin{equation}
g_i(\bx)=\sum_{k=0}^{M-1}B_{ik}\psi_k^{L}(\bx),\forall i\in[1,M]
\label{local}  
\end{equation}  
and are thus by construction linearly independent but are so far only
defined up to the expansion matrix $\mathbf{B}$. We determine
$\mathbf{B}$ by imposing
that the localizing functions should be localized near only one minimum
$\hat{\bx}_i$ and vanish at all remaining minima,
i.e. $g_i(\hat{\bx}_j)\simeq\delta_{i,j}$. Let the inverse of
$\mathbf{B}$ be $\mathbf{B}^{-1}$, $\mathbf{B}^{-1}\mathbf{B}=\mathbbm{1}$. We
finally fix $g_i(\bx)$ by imposing the following resolution of
identity $\sum_{I=1}^Mg_i(\bx)=1$, which allows us to write
\begin{equation}
\psi_i^{L}(\bx)=\sum_{k=1}^MB^{-1}_{ik}g_k(\bx),\forall i\in[0,M-1],
\label{invert}  
\end{equation}
We now define the
time-dependent population of the localizing sites (i.e. basins)
\begin{equation}
c_i(t)\equiv\int g_i(\bx)G(\bx,t|\bx_0) d\bx.
  \label{cst}
\end{equation}
The fact that $g_j(\bx)$ decompose unity implies that the total site
population is conserved in time, i.e.
\begin{equation}
\sum_{i=1}^Mc_i(t)\equiv\sum_{i=1}^M\int g_i(\bx)G(\bx,t|\bx_0) d\bx=\int
\sum_{i=1}^Mg_i(\bx)G(\bx,t|\bx_0) d\bx=1,
  \label{cons}
\end{equation}
where we have used the fact that the integral and sum commute by
Fubini's theorem (note that we can write the sum as an integral with
respect to a counting measure).
The localizing functions are linearly independent but not
orthonormal. For this purpose we define the $M\times M$ superposition matrix
$\mathbf{S}$ with elements $S_{ij}\equiv\int
g_i(\bx)P_{1}^{\eeq}(\bx)g_j(\bx)d\bx$ such that we can re-write the
equilibrium site
population as
\begin{equation}
c_i^{\eeq}=\int g_i(\bx)P_{1}^{\eeq}(\bx)\sum_{j=1}^Mg_j(\bx) d\bx=\sum_{j=1}^MS_{ij}.
\end{equation}
We now define a projection operator projecting onto the space of
localizing functions
\begin{equation}
\hat{\mathrm{P}}q(\bx)\equiv\sum_{i=1}^Mq_ig_i(\bx), \quad q_i\equiv
\sum_{k=1}^MS^{-1}_{i,k}\int
g_i(\bx)P_{1}^{\eeq}(\bx)q(\bx)d\bx.
\end{equation}
The time evolution of site
populations then follows 
\begin{eqnarray}
\frac{dc_j(t)}{dt}&=&\int g_j(\bx)\partial_tG_{\TT}(\bx,t|\bx_0)d\bx=\int
g_j(\bx)\LL_1G_{\TT}(\bx,t|\bx_0)d\bx=\int
G_{\TT}(\bx,t|\bx_0)\LL_1^{\dagger} g_j(\bx)d\bx\\
&\equiv&\int
G_{\TT}(\bx,t|\bx_0)\hat{\mathrm{P}}\LL_1^{\dagger} g_j(\bx)d\bx=\sum_{k,i}c_k(t)S^{-1}_{k,i}\int
g_i(\by)\LL_1^{\dagger}
P_{1}^{\eeq}(\by)g_j(\by)d\by\equiv\sum_{k,i}c_k(t)S^{-1}_{k,i}\Gamma_{ij},
\label{elMas}
\end{eqnarray}
where in the second line we used the fact that $\LL_1^{\dagger}
g_j(\bx)$ already lies in the subspace of localizing functions
(because $\LL^{\dagger}_T\psi_k^{L}(\bx)=-\lambda_k\psi_k^{L}(\bx)$ and Eq.~(\ref{local})) and the
projection operator projects back onto said subspace. 
By defining $\mathbf{c}(t)=(c_i(t),\ldots,c_M(t))^T$
we recognize from Eq.~(\ref{elMas}) that the site populations obey the Markovian master equation
\begin{equation}
\frac{d}{dt}\mathbf{c}(t)=\mathbf{M}\mathbf{c}(t), \quad M_{jk}\equiv \sum_{i}S^{-1}_{k,i}\Gamma_{ij},
\label{master}  
\end{equation}
where it can be shown that the transition rates entering $\mathbf{M}$
obey detailed balance \cite{SMoro}.
It is obvious that $\mathbf{M}\mathbf{c}^{eq}=0$ and therefore an
equilibrated site-population does not lead to any inter-well dynamics.
The evolution upon a temperature quench from $\TT$ follows from the
evolution of the Green's function, i.e. $P_{\TT}(\bx,t)=\int
G_1(\bx,t|\bx_0)P_{\TT}^{\eeq}(\bx_0)d\bx_0$. Therefore, any quench
that will leave the site populations given the potential $U(\bx)$
and Fokker-Planck operator $\LL_1$ ($\LL_1^{\dagger}$ respectively) almost unaffected, i.e.
\begin{equation}
\mathbf{M}\mathbf{c}(0)\simeq 0, \quad {\rm where} \quad c_i(0)=\int g_i(\bx)P_{\TT}(\bx,t=0|\bx_0) d\bx\equiv\int g_i(\bx_0)P_{\TT}^{\eeq}(\bx_0) d\bx_0,
\end{equation}  
will lead to a faster uphill relaxation as a direct consequence of the
fact that the intra-well (i.e. in each individual well) relaxation is faster uphill. The above arguments can be
arranged in a form that is fully rigorous, but since the argumentation is essentially
straightforward, we do not find it necessary to do so.

\subsection{Small local modulations do not spoil the asymmetry}
As stated in the Letter, small local modulations of the potential ($\ll
k_{\rm B}T_{\eeq}$) do not affect the asymmetry as longs as the uphill
 quench is sufficiently small to assure that the modulation is $\lesssim
k_{\rm B}T^-$
.  Then the system relaxes
similarly as in a perfectly smooth single well. To demonstrate that this is indeed
the case we inspect the relaxation from equidistant quenches in the
potential  in Eq.~(\ref{4well}) with $\Delta_x=\Delta_y=2$ and
$x_0^2=y^2=0.4$ depicted in Fig.~\ref{2D_cherry}.
If, however, we make the quench too severe, such that the local
modulations of the potential effectively reach $|\Delta U(\bx)|\gtrsim
k_{\rm B}T^-$ the asymmetry would become violated and the curves will
eventually cross, rendering downhill relaxation faster.      
     
      \begin{figure}[!ht]
      \begin{center}
       \includegraphics[width=0.85\textwidth]{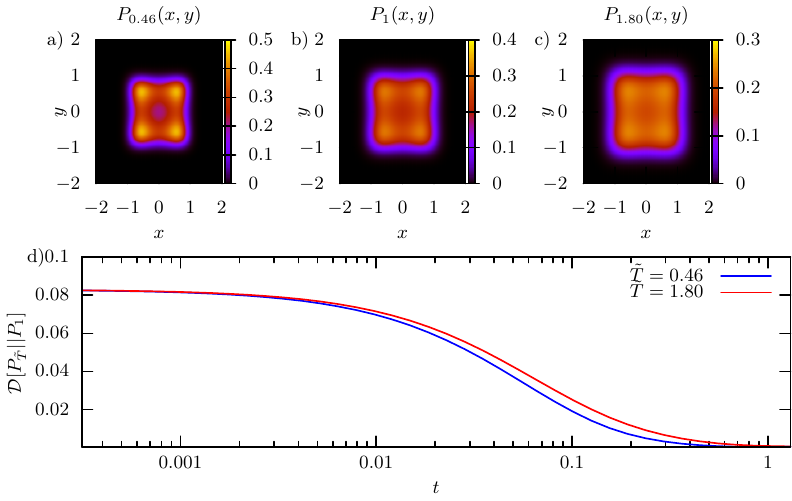}
      \end{center}
      \caption{a) Density of invariant measure at $\tilde{T}=1$
        (i.e. equilibrium probability density), and the equidistant post-quench
        probability densities at $\tilde{T}^+=1.8$ and
        $\tilde{T}^-=0.46$ for the 4-well potential in
        Eq.~(\ref{4well}) with parameters $\Delta_x=\Delta_y=2$ and $x_0^2=y_0^2=0.4$.; b) Corresponding time evolution of the
        Kullback-Leibler divergence depicting that the asymmetry is
        obeyed.}
      \label{2D_cherry}
     \end{figure}  

A a final example we focus on an asymmetric quadruple-well with a pair
of high barriers and a pair of low barriers (the latter creating a
small local modulation of the potential). In particular, we consider
the relaxation in the potential given in Eq.~(\ref{4well}) with
parameters $\Delta_x=3,\Delta_y=2$ and $x_0=0.5,y_0=1$ and inspect in
Fig.~\ref{2D_asym} the
following pairs of thermodynamically equidistant temperatures,
$\tilde{T}^-=0.8,\tilde{T}^+=1.25$ and $\tilde{T}^-=0.5,\tilde{T}^+=2.$.      
      
       \begin{figure}[!ht]
      \begin{center}
       \includegraphics[width=0.85\textwidth]{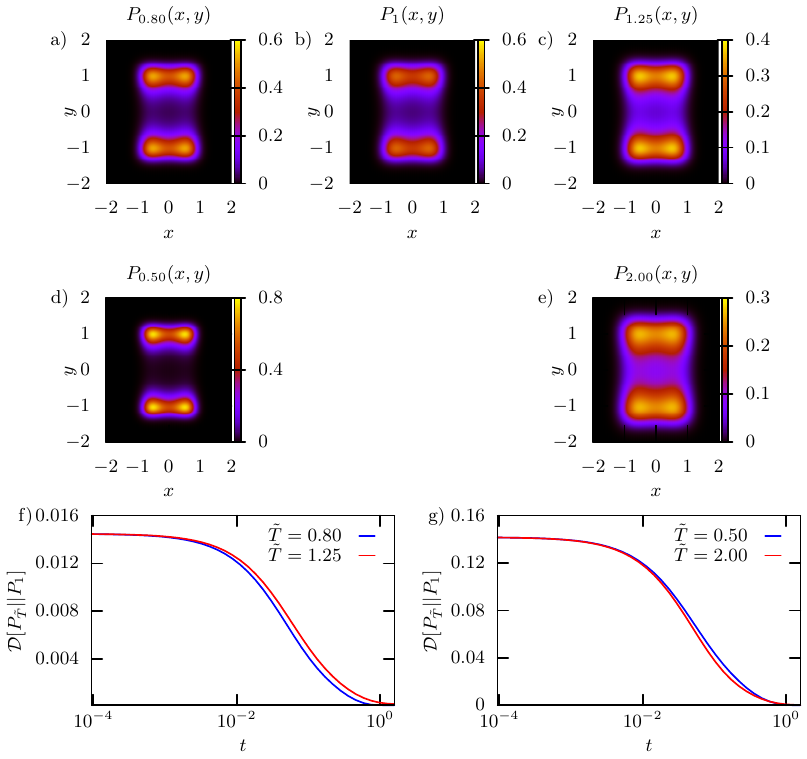}
      \end{center}
      \caption{b) Density of invariant measure at $\tilde{T}=1$
        (i.e. equilibrium probability density), and two pairs of equidistant post-quench
        probability densities at $\tilde{T}^+=1.25$ (c) and $2$ (e) and
        corresponding equidistant
        $\tilde{T}^-=0.8$ (a) and $0.5$ (d), respectively, for the 4-well potential in
        Eq.~(\ref{4well}) with parameters $\Delta_x=3,\Delta_y=2$ and $x_0=0.5,y_0=1$.; f-g) Corresponding time evolution of the
        Kullback-Leibler divergence depicting that the asymmetry is
        obeyed for small enough quenches but becomes violated (in the
        form of an Mpemba-like effect) for stronger quenches.}
      \label{2D_asym}
     \end{figure}     

As anticipated, the uphill relaxation is faster for sufficiently small
quenches (see  Fig.~\ref{2D_asym}f) and becomes violated for stronger
quenches (see  Fig.~\ref{2D_asym}g), where the Kullback-Leibler
divergences also display an Mpemba-like effect (see also next
section).

     \section{Generalized Mpemba effect for non-Markovian dynamics}
     A phenomenon closely linked to relaxation from a quench is the
     so-called Mpemba effect
     \cite{Se.b._mpemba_d.g._osborne_cool?_1979, Sjeng_mpemba_2006,
       Skatz_when_2009}, according to which a liquid upon cooling can
     freeze faster if its initial temperature is higher. Meanwhile the
     phenomenon has been extended to cover relaxation processes in
     different systems: magneto-resistors \cite{Schaddah_overtaking_2010}, carbon-nanotubes \cite{Sgreaney_mpemba-like_2011}, polymers crystallization \cite{Shu_conformation_2018}, clathrate hydrates \cite{Sahn_experimental_2016}, granular systems \cite{Slasanta_when_2017} and spin glasses \cite{Sjanus_collaboration_mpemba_2019}.  Recently theoretical generalizations of it for Markovian observables have been published \cite{Slu_nonequilibrium_2017, Sklich_solution_2019, Sklich_mpemba_2019}. Not long ago the phenomenon was also adressed in more detail in the
     context of Markovian stochastic dynamics
     \cite{Slu_nonequilibrium_2017,Sklich_mpemba_2019}. 

     Here we further extend the concept of the Mpemba effect to
     projected, non-Markovian observables. As before we focus on the
     distance of two different generic configurations displaced from
     equilibrium at $t=0$, such that one is displaced further away
     than the other, whereas the time-evolution of the entire system
     is governed by the same Fokker-Planck operator. In this setting,
     there are cases, where the more distant initial configuration
     reaches equilibrium faster that the closer one. One can observe
     this effect in the two systems analyzed in the Letter (see Fig. \ref{equilibration figure}).
     \begin{figure}[!ht]
     \begin{center}
     \includegraphics[width=0.5\textwidth]{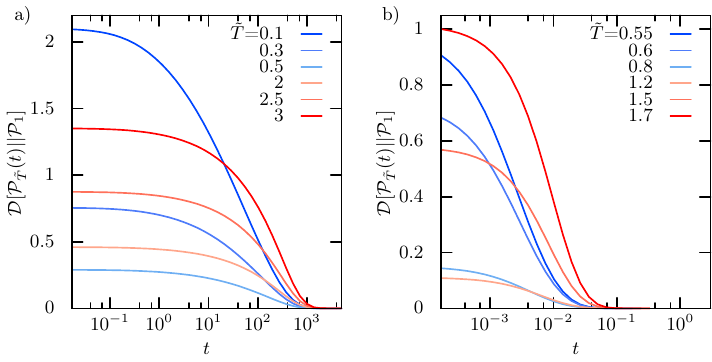}
      \caption{In the left panel we show time dependence of the
        Kullback-Leibler divergence for a Gaussian Chain of 100 beads,
        while the right panel depicts a Single File of 10 particles
        ($g=5$). In both cases we focus on non-Markovian observables,
        the end-to-end distance for the Gaussian chain and on the 7th particle of the single file, respectively. For some pairs of initial temperatures we notice the generalized Mpemba effect: systems that start further away from the equilibrium approach the equilibrium configuration faster.}
   \label{equilibration figure}
   \end{center}
  \end{figure}
    It is worth to stress that the presence of the generalized Mpemba
    effect not only depends on the system and the initial condition
    (like in the Markovian case) but also on the particular type of
    projection In Fig. \ref{mpembaondemand} we demonstrate, on hand of
    the same system (a tilted single file of 5 particles) from the
    same pair of pre-quench temperatures, that we can switch the
    generalized Mpemba effect on and off by simply changing the particle we are tagging.
     \begin{figure}[!ht]
     \begin{center}
      \includegraphics[width=0.5\textwidth]{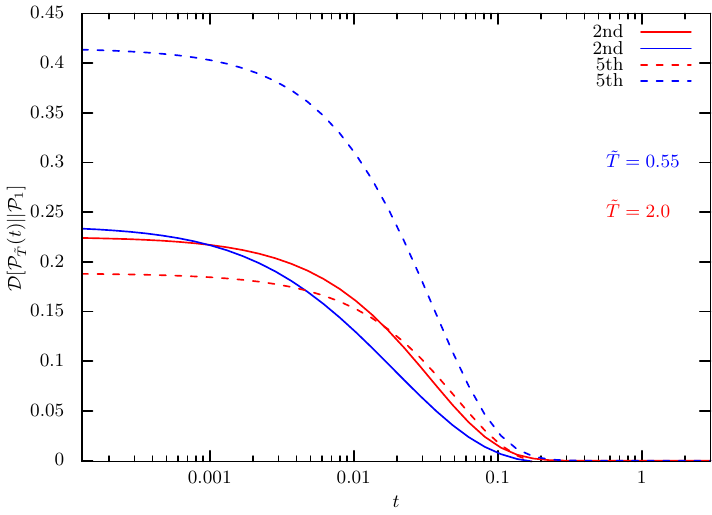}
      \caption{Kullback-Leibler divergences for a single file of 5
        particles with $g=1$. If we tag the 2nd particle (solid lines)
        or the 5th (dashed lines) for the same pair of pre-quench
        temperatures one projection displays the generalized Mpemba
        effect while the other one does not.}
      \label{mpembaondemand}
     \end{center}
     \end{figure}

     \newpage     

\begin{thebibliography}{56}
\makeatletter
\providecommand \@ifxundefined [1]{%
 \@ifx{#1\undefined}
}%
\providecommand \@ifnum [1]{%
 \ifnum #1\expandafter \@firstoftwo
 \else \expandafter \@secondoftwo
 \fi
}%
\providecommand \@ifx [1]{%
 \ifx #1\expandafter \@firstoftwo
 \else \expandafter \@secondoftwo
 \fi
}%
\providecommand \natexlab [1]{#1}%
\providecommand \enquote  [1]{``#1''}%
\providecommand \bibnamefont  [1]{#1}%
\providecommand \bibfnamefont [1]{#1}%
\providecommand \citenamefont [1]{#1}%
\providecommand \href@noop [0]{\@secondoftwo}%
\providecommand \href [0]{\begingroup \@sanitize@url \@href}%
\providecommand \@href[1]{\@@startlink{#1}\@@href}%
\providecommand \@@href[1]{\endgroup#1\@@endlink}%
\providecommand \@sanitize@url [0]{\catcode `\\12\catcode `\$12\catcode
  `\&12\catcode `\#12\catcode `\^12\catcode `\_12\catcode `\%12\relax}%
\providecommand \@@startlink[1]{}%
\providecommand \@@endlink[0]{}%
\providecommand \url  [0]{\begingroup\@sanitize@url \@url }%
\providecommand \@url [1]{\endgroup\@href {#1}{\urlprefix }}%
\providecommand \urlprefix  [0]{URL }%
\providecommand \Eprint [0]{\href }%
\providecommand \doibase [0]{http://dx.doi.org/}%
\providecommand \selectlanguage [0]{\@gobble}%
\providecommand \bibinfo  [0]{\@secondoftwo}%
\providecommand \bibfield  [0]{\@secondoftwo}%
\providecommand \translation [1]{[#1]}%
\providecommand \BibitemOpen [0]{}%
\providecommand \bibitemStop [0]{}%
\providecommand \bibitemNoStop [0]{.\EOS\space}%
\providecommand \EOS [0]{\spacefactor3000\relax}%
\providecommand \BibitemShut  [1]{\csname bibitem#1\endcsname}%
\let\auto@bib@innerbib\@empty
\bibitem [{\citenamefont {Farhan}\ \emph {et~al.}(2013)\citenamefont {Farhan},
  \citenamefont {Derlet}, \citenamefont {Kleibert}, \citenamefont {Balan},
  \citenamefont {Chopdekar}, \citenamefont {Wyss}, \citenamefont {Perron},
  \citenamefont {Scholl}, \citenamefont {Nolting},\ and\ \citenamefont
  {Heyderman}}]{farhan_direct_2013}%
  \BibitemOpen
  \bibfield  {author} {\bibinfo {author} {\bibfnamefont {A.}~\bibnamefont
  {Farhan}}, \bibinfo {author} {\bibfnamefont {P.~M.}\ \bibnamefont {Derlet}},
  \bibinfo {author} {\bibfnamefont {A.}~\bibnamefont {Kleibert}}, \bibinfo
  {author} {\bibfnamefont {A.}~\bibnamefont {Balan}}, \bibinfo {author}
  {\bibfnamefont {R.~V.}\ \bibnamefont {Chopdekar}}, \bibinfo {author}
  {\bibfnamefont {M.}~\bibnamefont {Wyss}}, \bibinfo {author} {\bibfnamefont
  {J.}~\bibnamefont {Perron}}, \bibinfo {author} {\bibfnamefont
  {A.}~\bibnamefont {Scholl}}, \bibinfo {author} {\bibfnamefont
  {F.}~\bibnamefont {Nolting}}, \ and\ \bibinfo {author} {\bibfnamefont
  {L.~J.}\ \bibnamefont {Heyderman}},\ }\href {\doibase
  10.1103/PhysRevLett.111.057204} {\bibfield  {journal} {\bibinfo  {journal}
  {Phys. Rev. Lett.}\ }\textbf {\bibinfo {volume} {111}},\ \bibinfo {pages}
  {057204} (\bibinfo {year} {2013})}\BibitemShut {NoStop}%
\bibitem [{\citenamefont {Dattagupta}(2012)}]{dattagupta_relaxation_2012}%
  \BibitemOpen
  \bibfield  {author} {\bibinfo {author} {\bibfnamefont {S.}~\bibnamefont
  {Dattagupta}},\ }\href@noop {} {\emph {\bibinfo {title} {Relaxation phenomena
  in condensed matter physics}}}\ (\bibinfo  {publisher} {Elsevier},\ \bibinfo
  {year} {2012})\BibitemShut {NoStop}%
\bibitem [{\citenamefont {Chen}\ \emph {et~al.}(2007)\citenamefont {Chen},
  \citenamefont {Rhoades}, \citenamefont {Butler}, \citenamefont {Loh},\ and\
  \citenamefont {Webb}}]{Chen_2007}%
  \BibitemOpen
  \bibfield  {author} {\bibinfo {author} {\bibfnamefont {H.}~\bibnamefont
  {Chen}}, \bibinfo {author} {\bibfnamefont {E.}~\bibnamefont {Rhoades}},
  \bibinfo {author} {\bibfnamefont {J.~S.}\ \bibnamefont {Butler}}, \bibinfo
  {author} {\bibfnamefont {S.~N.}\ \bibnamefont {Loh}}, \ and\ \bibinfo
  {author} {\bibfnamefont {W.~W.}\ \bibnamefont {Webb}},\ }\href {\doibase
  10.1073/pnas.0704073104} {\bibfield  {journal} {\bibinfo  {journal} {Proc.
  Natl. Acad. Sci. USA}\ }\textbf {\bibinfo {volume} {104}},\ \bibinfo {pages}
  {10459–10464} (\bibinfo {year} {2007})}\BibitemShut {NoStop}%
\bibitem [{\citenamefont {Wang}\ \emph {et~al.}(2006)\citenamefont {Wang},
  \citenamefont {Song},\ and\ \citenamefont {Makse}}]{Wang_2006}%
  \BibitemOpen
  \bibfield  {author} {\bibinfo {author} {\bibfnamefont {P.}~\bibnamefont
  {Wang}}, \bibinfo {author} {\bibfnamefont {C.}~\bibnamefont {Song}}, \ and\
  \bibinfo {author} {\bibfnamefont {H.~A.}\ \bibnamefont {Makse}},\ }\href
  {\doibase 10.1038/nphys366} {\bibfield  {journal} {\bibinfo  {journal} {Nat.
  Phys.}\ }\textbf {\bibinfo {volume} {2}},\ \bibinfo {pages} {526–531}
  (\bibinfo {year} {2006})}\BibitemShut {NoStop}%
\bibitem [{\citenamefont {Lizana}\ and\ \citenamefont
  {Ambjörnsson}(2008)}]{lizana_single-file_2008}%
  \BibitemOpen
  \bibfield  {author} {\bibinfo {author} {\bibfnamefont {L.}~\bibnamefont
  {Lizana}}\ and\ \bibinfo {author} {\bibfnamefont {T.}~\bibnamefont
  {Ambjörnsson}},\ }\href {\doibase 10.1103/PhysRevLett.100.200601} {\bibfield
   {journal} {\bibinfo  {journal} {Phys. Rev. Lett.}\ }\textbf {\bibinfo
  {volume} {100}},\ \bibinfo {pages} {200601} (\bibinfo {year}
  {2008})}\BibitemShut {NoStop}%
\bibitem [{\citenamefont {Lizana}\ and\ \citenamefont
  {Ambjörnsson}(2009)}]{lizana_diffusion_2009}%
  \BibitemOpen
  \bibfield  {author} {\bibinfo {author} {\bibfnamefont {L.}~\bibnamefont
  {Lizana}}\ and\ \bibinfo {author} {\bibfnamefont {T.}~\bibnamefont
  {Ambjörnsson}},\ }\href {\doibase 10.1103/PhysRevE.80.051103} {\bibfield
  {journal} {\bibinfo  {journal} {Phys. Rev. E}\ }\textbf {\bibinfo {volume}
  {80}},\ \bibinfo {pages} {051103} (\bibinfo {year} {2009})}\BibitemShut
  {NoStop}%
\bibitem [{\citenamefont {Lapolla}\ and\ \citenamefont
  {Godec}(2018)}]{Lapolla_2018}%
  \BibitemOpen
  \bibfield  {author} {\bibinfo {author} {\bibfnamefont {A.}~\bibnamefont
  {Lapolla}}\ and\ \bibinfo {author} {\bibfnamefont {A.}~\bibnamefont
  {Godec}},\ }\href {\doibase 10.1088/1367-2630/aaea1b} {\bibfield  {journal}
  {\bibinfo  {journal} {New J. Phys.}\ }\textbf {\bibinfo {volume} {20}},\
  \bibinfo {pages} {113021} (\bibinfo {year} {2018})}\BibitemShut {NoStop}%
\bibitem [{\citenamefont {Lapolla}\ and\ \citenamefont
  {Godec}(2019)}]{lapolla_manifestations_2019}%
  \BibitemOpen
  \bibfield  {author} {\bibinfo {author} {\bibfnamefont {A.}~\bibnamefont
  {Lapolla}}\ and\ \bibinfo {author} {\bibfnamefont {A.}~\bibnamefont
  {Godec}},\ }\href {\doibase 10.3389/fphy.2019.00182} {\bibfield  {journal}
  {\bibinfo  {journal} {Front. Phys.}\ }\textbf {\bibinfo {volume} {7}}
  (\bibinfo {year} {2019}),\ 10.3389/fphy.2019.00182}\BibitemShut {NoStop}%
\bibitem [{\citenamefont
  {Onsager}(1931{\natexlab{a}})}]{onsager_reciprocal_1931}%
  \BibitemOpen
  \bibfield  {author} {\bibinfo {author} {\bibfnamefont {L.}~\bibnamefont
  {Onsager}},\ }\href {\doibase 10.1103/PhysRev.37.405} {\bibfield  {journal}
  {\bibinfo  {journal} {Phys. Rev.}\ }\textbf {\bibinfo {volume} {37}},\
  \bibinfo {pages} {405} (\bibinfo {year} {1931}{\natexlab{a}})}\BibitemShut
  {NoStop}%
\bibitem [{\citenamefont
  {Onsager}(1931{\natexlab{b}})}]{onsager_reciprocal_1931-1}%
  \BibitemOpen
  \bibfield  {author} {\bibinfo {author} {\bibfnamefont {L.}~\bibnamefont
  {Onsager}},\ }\href {\doibase 10.1103/PhysRev.38.2265} {\bibfield  {journal}
  {\bibinfo  {journal} {Phys. Rev.}\ }\textbf {\bibinfo {volume} {38}},\
  \bibinfo {pages} {2265} (\bibinfo {year} {1931}{\natexlab{b}})}\BibitemShut
  {NoStop}%
\bibitem [{\citenamefont {Kubo}\ \emph {et~al.}(1957)\citenamefont {Kubo},
  \citenamefont {Yokota},\ and\ \citenamefont {Nakajima}}]{Kubo_1957}%
  \BibitemOpen
  \bibfield  {author} {\bibinfo {author} {\bibfnamefont {R.}~\bibnamefont
  {Kubo}}, \bibinfo {author} {\bibfnamefont {M.}~\bibnamefont {Yokota}}, \ and\
  \bibinfo {author} {\bibfnamefont {S.}~\bibnamefont {Nakajima}},\ }\href
  {\doibase 10.1143/jpsj.12.1203} {\bibfield  {journal} {\bibinfo  {journal}
  {J. Phys. Soc. Jpn.}\ }\textbf {\bibinfo {volume} {12}},\ \bibinfo {pages}
  {1203–1211} (\bibinfo {year} {1957})}\BibitemShut {NoStop}%
\bibitem [{\citenamefont {Seifert}(2012)}]{seifert_stochastic_2012}%
  \BibitemOpen
  \bibfield  {author} {\bibinfo {author} {\bibfnamefont {U.}~\bibnamefont
  {Seifert}},\ }\href {\doibase 10.1088/0034-4885/75/12/126001} {\bibfield
  {journal} {\bibinfo  {journal} {Rep. Prog. Phys.}\ }\textbf {\bibinfo
  {volume} {75}},\ \bibinfo {pages} {126001} (\bibinfo {year}
  {2012})}\BibitemShut {NoStop}%
\bibitem [{\citenamefont {Jarzynski}(2011)}]{jarzynski_equalities_2011}%
  \BibitemOpen
  \bibfield  {author} {\bibinfo {author} {\bibfnamefont {C.}~\bibnamefont
  {Jarzynski}},\ }\href {\doibase 10.1146/annurev-conmatphys-062910-140506}
  {\bibfield  {journal} {\bibinfo  {journal} {Annu. Rev. Condens. Matter
  Phys.}\ }\textbf {\bibinfo {volume} {2}},\ \bibinfo {pages} {329} (\bibinfo
  {year} {2011})}\BibitemShut {NoStop}%
\bibitem [{\citenamefont {Seifert}(2016)}]{Seifert_strong}%
  \BibitemOpen
  \bibfield  {author} {\bibinfo {author} {\bibfnamefont {U.}~\bibnamefont
  {Seifert}},\ }\href {\doibase 10.1103/PhysRevLett.116.020601} {\bibfield
  {journal} {\bibinfo  {journal} {Phys. Rev. Lett.}\ }\textbf {\bibinfo
  {volume} {116}},\ \bibinfo {pages} {020601} (\bibinfo {year}
  {2016})}\BibitemShut {NoStop}%
\bibitem [{\citenamefont {Strasberg}\ \emph {et~al.}(2016)\citenamefont
  {Strasberg}, \citenamefont {Schaller}, \citenamefont {Lambert},\ and\
  \citenamefont {Brandes}}]{Strasberg_2016}%
  \BibitemOpen
  \bibfield  {author} {\bibinfo {author} {\bibfnamefont {P.}~\bibnamefont
  {Strasberg}}, \bibinfo {author} {\bibfnamefont {G.}~\bibnamefont {Schaller}},
  \bibinfo {author} {\bibfnamefont {N.}~\bibnamefont {Lambert}}, \ and\
  \bibinfo {author} {\bibfnamefont {T.}~\bibnamefont {Brandes}},\ }\href
  {\doibase 10.1088/1367-2630/18/7/073007} {\bibfield  {journal} {\bibinfo
  {journal} {New J. Phys.}\ }\textbf {\bibinfo {volume} {18}},\ \bibinfo
  {pages} {073007} (\bibinfo {year} {2016})}\BibitemShut {NoStop}%
\bibitem [{\citenamefont {Strasberg}\ and\ \citenamefont
  {Esposito}(2017)}]{Massi}%
  \BibitemOpen
  \bibfield  {author} {\bibinfo {author} {\bibfnamefont {P.}~\bibnamefont
  {Strasberg}}\ and\ \bibinfo {author} {\bibfnamefont {M.}~\bibnamefont
  {Esposito}},\ }\href {\doibase 10.1103/PhysRevE.95.062101} {\bibfield
  {journal} {\bibinfo  {journal} {Phys. Rev. E}\ }\textbf {\bibinfo {volume}
  {95}},\ \bibinfo {pages} {062101} (\bibinfo {year} {2017})}\BibitemShut
  {NoStop}%
\bibitem [{\citenamefont {Jarzynski}(2017)}]{Jarzynski_strong}%
  \BibitemOpen
  \bibfield  {author} {\bibinfo {author} {\bibfnamefont {C.}~\bibnamefont
  {Jarzynski}},\ }\href {\doibase 10.1103/PhysRevX.7.011008} {\bibfield
  {journal} {\bibinfo  {journal} {Phys. Rev. X}\ }\textbf {\bibinfo {volume}
  {7}},\ \bibinfo {pages} {011008} (\bibinfo {year} {2017})}\BibitemShut
  {NoStop}%
\bibitem [{\citenamefont {Talkner}\ and\ \citenamefont
  {Hänggi}(2019)}]{talkner2019}%
  \BibitemOpen
  \bibfield  {author} {\bibinfo {author} {\bibfnamefont {P.}~\bibnamefont
  {Talkner}}\ and\ \bibinfo {author} {\bibfnamefont {P.}~\bibnamefont
  {Hänggi}},\ }\href@noop {} {\enquote {\bibinfo {title} {{\it Colloquium:}
  statistical mechanics and thermodynamics at strong coupling: Quantum and
  classical},}\ } (\bibinfo {year} {2019}),\ \Eprint
  {http://arxiv.org/abs/1911.11660} {arXiv:1911.11660} \BibitemShut {NoStop}%
\bibitem [{\citenamefont {Metzler}\ \emph {et~al.}(1999)\citenamefont
  {Metzler}, \citenamefont {Barkai},\ and\ \citenamefont {Klafter}}]{Ralf_1}%
  \BibitemOpen
  \bibfield  {author} {\bibinfo {author} {\bibfnamefont {R.}~\bibnamefont
  {Metzler}}, \bibinfo {author} {\bibfnamefont {E.}~\bibnamefont {Barkai}}, \
  and\ \bibinfo {author} {\bibfnamefont {J.}~\bibnamefont {Klafter}},\ }\href
  {\doibase 10.1103/PhysRevLett.82.3563} {\bibfield  {journal} {\bibinfo
  {journal} {Phys. Rev. Lett.}\ }\textbf {\bibinfo {volume} {82}},\ \bibinfo
  {pages} {3563} (\bibinfo {year} {1999})}\BibitemShut {NoStop}%
\bibitem [{\citenamefont {Metzler}\ and\ \citenamefont
  {Klafter}(2000)}]{Ralf_2}%
  \BibitemOpen
  \bibfield  {author} {\bibinfo {author} {\bibfnamefont {R.}~\bibnamefont
  {Metzler}}\ and\ \bibinfo {author} {\bibfnamefont {J.}~\bibnamefont
  {Klafter}},\ }\href {\doibase https://doi.org/10.1016/S0370-1573(00)00070-3}
  {\bibfield  {journal} {\bibinfo  {journal} {Phys. Rep.}\ }\textbf {\bibinfo
  {volume} {339}},\ \bibinfo {pages} {1 } (\bibinfo {year} {2000})}\BibitemShut
  {NoStop}%
\bibitem [{\citenamefont {Jeon}\ and\ \citenamefont {Metzler}(2012)}]{Ralf_3}%
  \BibitemOpen
  \bibfield  {author} {\bibinfo {author} {\bibfnamefont {J.-H.}\ \bibnamefont
  {Jeon}}\ and\ \bibinfo {author} {\bibfnamefont {R.}~\bibnamefont {Metzler}},\
  }\href {\doibase 10.1103/PhysRevE.85.021147} {\bibfield  {journal} {\bibinfo
  {journal} {Phys. Rev. E}\ }\textbf {\bibinfo {volume} {85}},\ \bibinfo
  {pages} {021147} (\bibinfo {year} {2012})}\BibitemShut {NoStop}%
\bibitem [{\citenamefont {Sokolov}\ and\ \citenamefont {Klafter}(2005)}]{Igor}%
  \BibitemOpen
  \bibfield  {author} {\bibinfo {author} {\bibfnamefont {I.~M.}\ \bibnamefont
  {Sokolov}}\ and\ \bibinfo {author} {\bibfnamefont {J.}~\bibnamefont
  {Klafter}},\ }\href {\doibase 10.1063/1.1860472} {\bibfield  {journal}
  {\bibinfo  {journal} {Chaos}\ }\textbf {\bibinfo {volume} {15}},\ \bibinfo
  {pages} {026103} (\bibinfo {year} {2005})}\BibitemShut {NoStop}%
\bibitem [{\citenamefont {Baiesi}\ and\ \citenamefont
  {Maes}(2013)}]{Baiesi_2013}%
  \BibitemOpen
  \bibfield  {author} {\bibinfo {author} {\bibfnamefont {M.}~\bibnamefont
  {Baiesi}}\ and\ \bibinfo {author} {\bibfnamefont {C.}~\bibnamefont {Maes}},\
  }\href {\doibase 10.1088/1367-2630/15/1/013004} {\bibfield  {journal}
  {\bibinfo  {journal} {New J. Phys.}\ }\textbf {\bibinfo {volume} {15}},\
  \bibinfo {pages} {013004} (\bibinfo {year} {2013})}\BibitemShut {NoStop}%
\bibitem [{\citenamefont {Maes}(2017)}]{Maes_2017}%
  \BibitemOpen
  \bibfield  {author} {\bibinfo {author} {\bibfnamefont {C.}~\bibnamefont
  {Maes}},\ }\href {\doibase 10.1103/physrevlett.119.160601} {\bibfield
  {journal} {\bibinfo  {journal} {Phys. Rev. Lett.}\ }\textbf {\bibinfo
  {volume} {119}},\ \bibinfo {pages} {160601} (\bibinfo {year}
  {2017})}\BibitemShut {NoStop}%
\bibitem [{\citenamefont {Blickle}\ \emph {et~al.}(2006)\citenamefont
  {Blickle}, \citenamefont {Speck}, \citenamefont {Helden}, \citenamefont
  {Seifert},\ and\ \citenamefont {Bechinger}}]{Blickle_2006}%
  \BibitemOpen
  \bibfield  {author} {\bibinfo {author} {\bibfnamefont {V.}~\bibnamefont
  {Blickle}}, \bibinfo {author} {\bibfnamefont {T.}~\bibnamefont {Speck}},
  \bibinfo {author} {\bibfnamefont {L.}~\bibnamefont {Helden}}, \bibinfo
  {author} {\bibfnamefont {U.}~\bibnamefont {Seifert}}, \ and\ \bibinfo
  {author} {\bibfnamefont {C.}~\bibnamefont {Bechinger}},\ }\href {\doibase
  10.1103/physrevlett.96.070603} {\bibfield  {journal} {\bibinfo  {journal}
  {Phys. Rev. Lett.}\ }\textbf {\bibinfo {volume} {96}},\ \bibinfo {pages}
  {070603} (\bibinfo {year} {2006})}\BibitemShut {NoStop}%
\bibitem [{\citenamefont {Hoang}\ \emph {et~al.}(2018)\citenamefont {Hoang},
  \citenamefont {Pan}, \citenamefont {Ahn}, \citenamefont {Bang}, \citenamefont
  {Quan},\ and\ \citenamefont {Li}}]{Hoang_2018}%
  \BibitemOpen
  \bibfield  {author} {\bibinfo {author} {\bibfnamefont {T.~M.}\ \bibnamefont
  {Hoang}}, \bibinfo {author} {\bibfnamefont {R.}~\bibnamefont {Pan}}, \bibinfo
  {author} {\bibfnamefont {J.}~\bibnamefont {Ahn}}, \bibinfo {author}
  {\bibfnamefont {J.}~\bibnamefont {Bang}}, \bibinfo {author} {\bibfnamefont
  {H.}~\bibnamefont {Quan}}, \ and\ \bibinfo {author} {\bibfnamefont
  {T.}~\bibnamefont {Li}},\ }\href {\doibase 10.1103/physrevlett.120.080602}
  {\bibfield  {journal} {\bibinfo  {journal} {Phys. Rev. Lett.}\ }\textbf
  {\bibinfo {volume} {120}},\ \bibinfo {pages} {080602} (\bibinfo {year}
  {2018})}\BibitemShut {NoStop}%
\bibitem [{\citenamefont {Jeon}\ \emph {et~al.}(2013)\citenamefont {Jeon},
  \citenamefont {Leijnse}, \citenamefont {Oddershede},\ and\ \citenamefont
  {Metzler}}]{Jeon_2013}%
  \BibitemOpen
  \bibfield  {author} {\bibinfo {author} {\bibfnamefont {J.-H.}\ \bibnamefont
  {Jeon}}, \bibinfo {author} {\bibfnamefont {N.}~\bibnamefont {Leijnse}},
  \bibinfo {author} {\bibfnamefont {L.~B.}\ \bibnamefont {Oddershede}}, \ and\
  \bibinfo {author} {\bibfnamefont {R.}~\bibnamefont {Metzler}},\ }\href
  {\doibase 10.1088/1367-2630/15/4/045011} {\bibfield  {journal} {\bibinfo
  {journal} {New J. Phys.}\ }\textbf {\bibinfo {volume} {15}},\ \bibinfo
  {pages} {045011} (\bibinfo {year} {2013})}\BibitemShut {NoStop}%
\bibitem [{\citenamefont {Collin}\ \emph {et~al.}(2005)\citenamefont {Collin},
  \citenamefont {Ritort}, \citenamefont {Jarzynski}, \citenamefont {Smith},
  \citenamefont {Tinoco},\ and\ \citenamefont {Bustamante}}]{Collin_2005}%
  \BibitemOpen
  \bibfield  {author} {\bibinfo {author} {\bibfnamefont {D.}~\bibnamefont
  {Collin}}, \bibinfo {author} {\bibfnamefont {F.}~\bibnamefont {Ritort}},
  \bibinfo {author} {\bibfnamefont {C.}~\bibnamefont {Jarzynski}}, \bibinfo
  {author} {\bibfnamefont {S.~B.}\ \bibnamefont {Smith}}, \bibinfo {author}
  {\bibfnamefont {I.}~\bibnamefont {Tinoco}}, \ and\ \bibinfo {author}
  {\bibfnamefont {C.}~\bibnamefont {Bustamante}},\ }\href {\doibase
  10.1038/nature04061} {\bibfield  {journal} {\bibinfo  {journal} {Nature}\
  }\textbf {\bibinfo {volume} {437}},\ \bibinfo {pages} {231–234} (\bibinfo
  {year} {2005})}\BibitemShut {NoStop}%
\bibitem [{\citenamefont {Dieterich}\ \emph {et~al.}(2015)\citenamefont
  {Dieterich}, \citenamefont {Camunas-Soler}, \citenamefont
  {Ribezzi-Crivellari}, \citenamefont {Seifert},\ and\ \citenamefont
  {Ritort}}]{Dieterich_2015}%
  \BibitemOpen
  \bibfield  {author} {\bibinfo {author} {\bibfnamefont {E.}~\bibnamefont
  {Dieterich}}, \bibinfo {author} {\bibfnamefont {J.}~\bibnamefont
  {Camunas-Soler}}, \bibinfo {author} {\bibfnamefont {M.}~\bibnamefont
  {Ribezzi-Crivellari}}, \bibinfo {author} {\bibfnamefont {U.}~\bibnamefont
  {Seifert}}, \ and\ \bibinfo {author} {\bibfnamefont {F.}~\bibnamefont
  {Ritort}},\ }\href {\doibase 10.1038/nphys3435} {\bibfield  {journal}
  {\bibinfo  {journal} {Nat. Phys.}\ }\textbf {\bibinfo {volume} {11}},\
  \bibinfo {pages} {971–977} (\bibinfo {year} {2015})}\BibitemShut {NoStop}%
\bibitem [{\citenamefont {Camunas-Soler}\ \emph {et~al.}(2017)\citenamefont
  {Camunas-Soler}, \citenamefont {Alemany},\ and\ \citenamefont
  {Ritort}}]{Camunas_Soler_2017}%
  \BibitemOpen
  \bibfield  {author} {\bibinfo {author} {\bibfnamefont {J.}~\bibnamefont
  {Camunas-Soler}}, \bibinfo {author} {\bibfnamefont {A.}~\bibnamefont
  {Alemany}}, \ and\ \bibinfo {author} {\bibfnamefont {F.}~\bibnamefont
  {Ritort}},\ }\href {\doibase 10.1126/science.aah4077} {\bibfield  {journal}
  {\bibinfo  {journal} {Science}\ }\textbf {\bibinfo {volume} {355}},\ \bibinfo
  {pages} {412–415} (\bibinfo {year} {2017})}\BibitemShut {NoStop}%
\bibitem [{\citenamefont {Hayashi}\ \emph {et~al.}(2010)\citenamefont
  {Hayashi}, \citenamefont {Ueno}, \citenamefont {Iino},\ and\ \citenamefont
  {Noji}}]{Hayashi_2010}%
  \BibitemOpen
  \bibfield  {author} {\bibinfo {author} {\bibfnamefont {K.}~\bibnamefont
  {Hayashi}}, \bibinfo {author} {\bibfnamefont {H.}~\bibnamefont {Ueno}},
  \bibinfo {author} {\bibfnamefont {R.}~\bibnamefont {Iino}}, \ and\ \bibinfo
  {author} {\bibfnamefont {H.}~\bibnamefont {Noji}},\ }\href {\doibase
  10.1103/physrevlett.104.218103} {\bibfield  {journal} {\bibinfo  {journal}
  {Phys. Rev. Lett.}\ }\textbf {\bibinfo {volume} {104}},\ \bibinfo {pages}
  {218103} (\bibinfo {year} {2010})}\BibitemShut {NoStop}%
\bibitem [{\citenamefont {Hatano}\ and\ \citenamefont {Sasa}(2001)}]{Sasa}%
  \BibitemOpen
  \bibfield  {author} {\bibinfo {author} {\bibfnamefont {T.}~\bibnamefont
  {Hatano}}\ and\ \bibinfo {author} {\bibfnamefont {S.-i.}\ \bibnamefont
  {Sasa}},\ }\href {\doibase 10.1103/PhysRevLett.86.3463} {\bibfield  {journal}
  {\bibinfo  {journal} {Phys. Rev. Lett.}\ }\textbf {\bibinfo {volume} {86}},\
  \bibinfo {pages} {3463} (\bibinfo {year} {2001})}\BibitemShut {NoStop}%
\bibitem [{\citenamefont {Maes}\ \emph {et~al.}(2011)\citenamefont {Maes},
  \citenamefont {Netočný},\ and\ \citenamefont {Wynants}}]{Maes_2011}%
  \BibitemOpen
  \bibfield  {author} {\bibinfo {author} {\bibfnamefont {C.}~\bibnamefont
  {Maes}}, \bibinfo {author} {\bibfnamefont {K.}~\bibnamefont {Netočný}}, \
  and\ \bibinfo {author} {\bibfnamefont {B.}~\bibnamefont {Wynants}},\ }\href
  {\doibase 10.1103/PhysRevLett.107.010601} {\bibfield  {journal} {\bibinfo
  {journal} {Phys. Rev. Lett.}\ }\textbf {\bibinfo {volume} {107}},\ \bibinfo
  {pages} {010601} (\bibinfo {year} {2011})}\BibitemShut {NoStop}%
\bibitem [{\citenamefont {Polettini}\ and\ \citenamefont
  {Esposito}(2013)}]{polettini_nonconvexity_2013}%
  \BibitemOpen
  \bibfield  {author} {\bibinfo {author} {\bibfnamefont {M.}~\bibnamefont
  {Polettini}}\ and\ \bibinfo {author} {\bibfnamefont {M.}~\bibnamefont
  {Esposito}},\ }\href {\doibase 10.1103/PhysRevE.88.012112} {\bibfield
  {journal} {\bibinfo  {journal} {Phys. Rev. E}\ }\textbf {\bibinfo {volume}
  {88}},\ \bibinfo {pages} {012112} (\bibinfo {year} {2013})}\BibitemShut
  {NoStop}%
\bibitem [{\citenamefont {Maes}(2020)}]{Maes_2019}%
  \BibitemOpen
  \bibfield  {author} {\bibinfo {author} {\bibfnamefont {C.}~\bibnamefont
  {Maes}},\ }\href {\doibase https://doi.org/10.1016/j.physrep.2020.01.002}
  {\bibfield  {journal} {\bibinfo  {journal} {Phys. Rep.}\ }\textbf {\bibinfo
  {volume} {850}},\ \bibinfo {pages} {1 } (\bibinfo {year} {2020})}\BibitemShut
  {NoStop}%
\bibitem [{\citenamefont {Speck}\ and\ \citenamefont {Seifert}(2005)}]{Speck}%
  \BibitemOpen
  \bibfield  {author} {\bibinfo {author} {\bibfnamefont {T.}~\bibnamefont
  {Speck}}\ and\ \bibinfo {author} {\bibfnamefont {U.}~\bibnamefont
  {Seifert}},\ }\href {\doibase 10.1088/0305-4470/38/34/l03} {\bibfield
  {journal} {\bibinfo  {journal} {J. Phys. A: Math. Gen.}\ }\textbf {\bibinfo
  {volume} {38}},\ \bibinfo {pages} {L581} (\bibinfo {year}
  {2005})}\BibitemShut {NoStop}%
\bibitem [{\citenamefont {Ch{\'{e}}trite}\ \emph {et~al.}(2019)\citenamefont
  {Ch{\'{e}}trite}, \citenamefont {Gupta}, \citenamefont {Neri},\ and\
  \citenamefont {Rold{\'{a}}n}}]{Roldan}%
  \BibitemOpen
  \bibfield  {author} {\bibinfo {author} {\bibfnamefont {R.}~\bibnamefont
  {Ch{\'{e}}trite}}, \bibinfo {author} {\bibfnamefont {S.}~\bibnamefont
  {Gupta}}, \bibinfo {author} {\bibfnamefont {I.}~\bibnamefont {Neri}}, \ and\
  \bibinfo {author} {\bibfnamefont {{\'{E}}.}~\bibnamefont {Rold{\'{a}}n}},\
  }\href {\doibase 10.1209/0295-5075/124/60006} {\bibfield  {journal} {\bibinfo
   {journal} {{EPL}}\ }\textbf {\bibinfo {volume} {124}},\ \bibinfo {pages}
  {60006} (\bibinfo {year} {2019})}\BibitemShut {NoStop}%
\bibitem [{\citenamefont {Chun}\ and\ \citenamefont {Noh}(2019)}]{Noh}%
  \BibitemOpen
  \bibfield  {author} {\bibinfo {author} {\bibfnamefont {H.-M.}\ \bibnamefont
  {Chun}}\ and\ \bibinfo {author} {\bibfnamefont {J.~D.}\ \bibnamefont {Noh}},\
  }\href {\doibase 10.1103/PhysRevE.99.012136} {\bibfield  {journal} {\bibinfo
  {journal} {Phys. Rev. E}\ }\textbf {\bibinfo {volume} {99}},\ \bibinfo
  {pages} {012136} (\bibinfo {year} {2019})}\BibitemShut {NoStop}%
\bibitem [{\citenamefont {Lu}\ and\ \citenamefont
  {Raz}(2017)}]{lu_nonequilibrium_2017}%
  \BibitemOpen
  \bibfield  {author} {\bibinfo {author} {\bibfnamefont {Z.}~\bibnamefont
  {Lu}}\ and\ \bibinfo {author} {\bibfnamefont {O.}~\bibnamefont {Raz}},\
  }\href {\doibase 10.1073/pnas.1701264114} {\bibfield  {journal} {\bibinfo
  {journal} {Proc. Natl. Acad. Sci. USA}\ }\textbf {\bibinfo {volume} {114}},\
  \bibinfo {pages} {5083} (\bibinfo {year} {2017})}\BibitemShut {NoStop}%
\bibitem [{\citenamefont {Klich}\ \emph {et~al.}(2019)\citenamefont {Klich},
  \citenamefont {Raz}, \citenamefont {Hirschberg},\ and\ \citenamefont
  {Vucelja}}]{klich_mpemba_2019}%
  \BibitemOpen
  \bibfield  {author} {\bibinfo {author} {\bibfnamefont {I.}~\bibnamefont
  {Klich}}, \bibinfo {author} {\bibfnamefont {O.}~\bibnamefont {Raz}}, \bibinfo
  {author} {\bibfnamefont {O.}~\bibnamefont {Hirschberg}}, \ and\ \bibinfo
  {author} {\bibfnamefont {M.}~\bibnamefont {Vucelja}},\ }\href {\doibase
  10.1103/PhysRevX.9.021060} {\bibfield  {journal} {\bibinfo  {journal} {Phys.
  Rev. X}\ }\textbf {\bibinfo {volume} {9}},\ \bibinfo {pages} {021060}
  (\bibinfo {year} {2019})}\BibitemShut {NoStop}%
\bibitem [{\citenamefont {Shiraishi}\ and\ \citenamefont
  {Saito}(2019)}]{shiraishi_information-theoretical_2019}%
  \BibitemOpen
  \bibfield  {author} {\bibinfo {author} {\bibfnamefont {N.}~\bibnamefont
  {Shiraishi}}\ and\ \bibinfo {author} {\bibfnamefont {K.}~\bibnamefont
  {Saito}},\ }\href {\doibase 10.1103/PhysRevLett.123.110603} {\bibfield
  {journal} {\bibinfo  {journal} {Phys. Rev. Lett.}\ }\textbf {\bibinfo
  {volume} {123}},\ \bibinfo {pages} {110603} (\bibinfo {year}
  {2019})}\BibitemShut {NoStop}%
\bibitem [{\citenamefont {Hartich}\ and\ \citenamefont {Godec}(2018)}]{David}%
  \BibitemOpen
  \bibfield  {author} {\bibinfo {author} {\bibfnamefont {D.}~\bibnamefont
  {Hartich}}\ and\ \bibinfo {author} {\bibfnamefont {A.}~\bibnamefont
  {Godec}},\ }\href {\doibase 10.1088/1367-2630/aaf038} {\bibfield  {journal}
  {\bibinfo  {journal} {New J. Phys.}\ }\textbf {\bibinfo {volume} {20}},\
  \bibinfo {pages} {112002} (\bibinfo {year} {2018})}\BibitemShut {NoStop}%
\bibitem [{\citenamefont {Hartich}\ and\ \citenamefont
  {Godec}(2019)}]{Hartich_2019}%
  \BibitemOpen
  \bibfield  {author} {\bibinfo {author} {\bibfnamefont {D.}~\bibnamefont
  {Hartich}}\ and\ \bibinfo {author} {\bibfnamefont {A.}~\bibnamefont
  {Godec}},\ }\href {\doibase 10.1088/1742-5468/ab00df} {\bibfield  {journal}
  {\bibinfo  {journal} {J. Stat. Mech. Theor. Exp.}\ }\textbf {\bibinfo
  {volume} {2019}},\ \bibinfo {pages} {024002} (\bibinfo {year}
  {2019})}\BibitemShut {NoStop}%
\bibitem [{\citenamefont {Mart{í}nez}\ \emph {et~al.}(2015)\citenamefont
  {Mart{í}nez}, \citenamefont {Rold{á}n}, \citenamefont {Dinis},
  \citenamefont {Petrov}, \citenamefont {Parrondo},\ and\ \citenamefont
  {Rica}}]{Edgar_NP}%
  \BibitemOpen
  \bibfield  {author} {\bibinfo {author} {\bibfnamefont {I.~A.}\ \bibnamefont
  {Mart{í}nez}}, \bibinfo {author} {\bibfnamefont {E.}~\bibnamefont
  {Rold{á}n}}, \bibinfo {author} {\bibfnamefont {L.}~\bibnamefont {Dinis}},
  \bibinfo {author} {\bibfnamefont {D.}~\bibnamefont {Petrov}}, \bibinfo
  {author} {\bibfnamefont {J.~M.~R.}\ \bibnamefont {Parrondo}}, \ and\ \bibinfo
  {author} {\bibfnamefont {R.~A.}\ \bibnamefont {Rica}},\ }\href {\doibase
  10.1038/nphys3518} {\bibfield  {journal} {\bibinfo  {journal} {Nat. Phys.}\
  }\textbf {\bibinfo {volume} {12}},\ \bibinfo {pages} {67–70} (\bibinfo
  {year} {2015})}\BibitemShut {NoStop}%
\bibitem [{\citenamefont {Mart\'{\i}nez}\ \emph {et~al.}(2013)\citenamefont
  {Mart\'{\i}nez}, \citenamefont {Rold\'an}, \citenamefont {Parrondo},\ and\
  \citenamefont {Petrov}}]{Edgar_PRE}%
  \BibitemOpen
  \bibfield  {author} {\bibinfo {author} {\bibfnamefont {I.~A.}\ \bibnamefont
  {Mart\'{\i}nez}}, \bibinfo {author} {\bibfnamefont {E.}~\bibnamefont
  {Rold\'an}}, \bibinfo {author} {\bibfnamefont {J.~M.~R.}\ \bibnamefont
  {Parrondo}}, \ and\ \bibinfo {author} {\bibfnamefont {D.}~\bibnamefont
  {Petrov}},\ }\href {\doibase 10.1103/PhysRevE.87.032159} {\bibfield
  {journal} {\bibinfo  {journal} {Phys. Rev. E}\ }\textbf {\bibinfo {volume}
  {87}},\ \bibinfo {pages} {032159} (\bibinfo {year} {2013})}\BibitemShut
  {NoStop}%
\bibitem [{\citenamefont {de Lorenzo}\ \emph {et~al.}(2015)\citenamefont
  {de Lorenzo}, \citenamefont {Ribezzi-Crivellari}, \citenamefont
  {Arias-Gonzalez}, \citenamefont {Smith},\ and\ \citenamefont
  {Ritort}}]{delorenzo_2015}%
  \BibitemOpen
  \bibfield  {author} {\bibinfo {author} {\bibfnamefont {S.}~\bibnamefont
  {de Lorenzo}}, \bibinfo {author} {\bibfnamefont {M.}~\bibnamefont
  {Ribezzi-Crivellari}}, \bibinfo {author} {\bibfnamefont {J.}~\bibnamefont
  {Arias-Gonzalez}}, \bibinfo {author} {\bibfnamefont {S.}~\bibnamefont
  {Smith}}, \ and\ \bibinfo {author} {\bibfnamefont {F.}~\bibnamefont
  {Ritort}},\ }\href {\doibase 10.1016/j.bpj.2015.05.017} {\bibfield  {journal}
  {\bibinfo  {journal} {Biophys. J.}\ }\textbf {\bibinfo {volume} {108}},\
  \bibinfo {pages} {2854} (\bibinfo {year} {2015})}\BibitemShut {NoStop}%
\bibitem [{\citenamefont {Gladrow}\ \emph {et~al.}(2019)\citenamefont
  {Gladrow}, \citenamefont {Ribezzi-Crivellari}, \citenamefont {Ritort},\ and\
  \citenamefont {Keyser}}]{Gladrow_2019}%
  \BibitemOpen
  \bibfield  {author} {\bibinfo {author} {\bibfnamefont {J.}~\bibnamefont
  {Gladrow}}, \bibinfo {author} {\bibfnamefont {M.}~\bibnamefont
  {Ribezzi-Crivellari}}, \bibinfo {author} {\bibfnamefont {F.}~\bibnamefont
  {Ritort}}, \ and\ \bibinfo {author} {\bibfnamefont {U.~F.}\ \bibnamefont
  {Keyser}},\ }\href {\doibase 10.1038/s41467-018-07873-9} {\bibfield
  {journal} {\bibinfo  {journal} {Nat. Commun.}\ }\textbf {\bibinfo {volume}
  {10}} (\bibinfo {year} {2019}),\ 10.1038/s41467-018-07873-9}\BibitemShut
  {NoStop}%
\bibitem [{\citenamefont {Sitters}\ \emph {et~al.}(2016)\citenamefont
  {Sitters}, \citenamefont {Laurens}, \citenamefont {de Rijk}, \citenamefont
  {Kress}, \citenamefont {Peterman},\ and\ \citenamefont {Wuite}}]{pushing}%
  \BibitemOpen
  \bibfield  {author} {\bibinfo {author} {\bibfnamefont {G.}~\bibnamefont
  {Sitters}}, \bibinfo {author} {\bibfnamefont {N.}~\bibnamefont {Laurens}},
  \bibinfo {author} {\bibfnamefont {E.}~\bibnamefont {de Rijk}}, \bibinfo
  {author} {\bibfnamefont {H.}~\bibnamefont {Kress}}, \bibinfo {author}
  {\bibfnamefont {E.}~\bibnamefont {Peterman}}, \ and\ \bibinfo {author}
  {\bibfnamefont {G.}~\bibnamefont {Wuite}},\ }\href {\doibase
  https://doi.org/10.1016/j.bpj.2015.11.028} {\bibfield  {journal} {\bibinfo
  {journal} {Biophys. J.}\ }\textbf {\bibinfo {volume} {110}},\ \bibinfo
  {pages} {44 } (\bibinfo {year} {2016})}\BibitemShut {NoStop}%
\bibitem [{\citenamefont {{S. Kullback}}\ and\ \citenamefont {{R.
  Leibler}}(1951)}]{s._kullback_information_1951}%
  \BibitemOpen
  \bibfield  {author} {\bibinfo {author} {\bibnamefont {{S. Kullback}}}\ and\
  \bibinfo {author} {\bibnamefont {{R. Leibler}}},\ }\href@noop {} {\bibfield
  {journal} {\bibinfo  {journal} {Ann. Math. Statist}\ }\textbf {\bibinfo
  {volume} {22}},\ \bibinfo {pages} {79} (\bibinfo {year} {1951})}\BibitemShut
  {NoStop}%
\bibitem [{\citenamefont {Lebowitz}\ and\ \citenamefont
  {Bergmann}(1957)}]{lebowitz_irreversible_1957}%
  \BibitemOpen
  \bibfield  {author} {\bibinfo {author} {\bibfnamefont {J.~L.}\ \bibnamefont
  {Lebowitz}}\ and\ \bibinfo {author} {\bibfnamefont {P.~G.}\ \bibnamefont
  {Bergmann}},\ }\href {\doibase 10.1016/0003-4916(57)90002-7} {\bibfield
  {journal} {\bibinfo  {journal} {Ann. Phys.}\ }\textbf {\bibinfo {volume}
  {1}},\ \bibinfo {pages} {1} (\bibinfo {year} {1957})}\BibitemShut {NoStop}%
\bibitem [{\citenamefont {Mackey}(1989)}]{Mackey_1989}%
  \BibitemOpen
  \bibfield  {author} {\bibinfo {author} {\bibfnamefont {M.~C.}\ \bibnamefont
  {Mackey}},\ }\href {\doibase 10.1103/revmodphys.61.981} {\bibfield  {journal}
  {\bibinfo  {journal} {Rev. Mod. Phys.}\ }\textbf {\bibinfo {volume} {61}},\
  \bibinfo {pages} {981–1015} (\bibinfo {year} {1989})}\BibitemShut {NoStop}%
\bibitem [{\citenamefont {Qian}(2013)}]{Qian_2013}%
  \BibitemOpen
  \bibfield  {author} {\bibinfo {author} {\bibfnamefont {H.}~\bibnamefont
  {Qian}},\ }\href {\doibase 10.1063/1.4803847} {\bibfield  {journal} {\bibinfo
   {journal} {J. Math. Phys.}\ }\textbf {\bibinfo {volume} {54}},\ \bibinfo
  {pages} {053302} (\bibinfo {year} {2013})}\BibitemShut {NoStop}%
\bibitem [{\citenamefont {Van~den Broeck}\ and\ \citenamefont
  {Esposito}(2010)}]{Massi_proof}%
  \BibitemOpen
  \bibfield  {author} {\bibinfo {author} {\bibfnamefont {C.}~\bibnamefont
  {Van~den Broeck}}\ and\ \bibinfo {author} {\bibfnamefont {M.}~\bibnamefont
  {Esposito}},\ }\href {\doibase 10.1103/PhysRevE.82.011144} {\bibfield
  {journal} {\bibinfo  {journal} {Phys. Rev. E}\ }\textbf {\bibinfo {volume}
  {82}},\ \bibinfo {pages} {011144} (\bibinfo {year} {2010})}\BibitemShut
  {NoStop}%
\bibitem [{\citenamefont {Esposito}\ and\ \citenamefont {Van~den
  Broeck}(2010)}]{Massi_PRL}%
  \BibitemOpen
  \bibfield  {author} {\bibinfo {author} {\bibfnamefont {M.}~\bibnamefont
  {Esposito}}\ and\ \bibinfo {author} {\bibfnamefont {C.}~\bibnamefont {Van~den
  Broeck}},\ }\href {\doibase 10.1103/PhysRevLett.104.090601} {\bibfield
  {journal} {\bibinfo  {journal} {Phys. Rev. Lett.}\ }\textbf {\bibinfo
  {volume} {104}},\ \bibinfo {pages} {090601} (\bibinfo {year}
  {2010})}\BibitemShut {NoStop}%
\bibitem [{\citenamefont {Vaikuntanathan}\ and\ \citenamefont
  {Jarzynski}(2009)}]{Vaiku}%
  \BibitemOpen
  \bibfield  {author} {\bibinfo {author} {\bibfnamefont {S.}~\bibnamefont
  {Vaikuntanathan}}\ and\ \bibinfo {author} {\bibfnamefont {C.}~\bibnamefont
  {Jarzynski}},\ }\href {\doibase 10.1209/0295-5075/87/60005} {\bibfield
  {journal} {\bibinfo  {journal} {{EPL}}\ }\textbf {\bibinfo {volume} {87}},\
  \bibinfo {pages} {60005} (\bibinfo {year} {2009})}\BibitemShut {NoStop}%
\bibitem [{\citenamefont {Doi}\ and\ \citenamefont
  {Edwards}(1988)}]{doi_theory_1988}%
  \BibitemOpen
  \bibfield  {author} {\bibinfo {author} {\bibfnamefont {M.}~\bibnamefont
  {Doi}}\ and\ \bibinfo {author} {\bibfnamefont {S.~F.}\ \bibnamefont
  {Edwards}},\ }\href@noop {} {\emph {\bibinfo {title} {The Theory of Polymer
  Dynamics}}}\ (\bibinfo  {publisher} {Clarendon Press},\ \bibinfo {year}
  {1988})\BibitemShut {NoStop}%
\bibitem [{Note1()}]{Note1}%
  \BibitemOpen
  \bibinfo {note} {$T>T_{\protect \mathrm {eq}}$ implies $\protect \mathaccentV
  {tilde}07E{T}>1$ and $T<T_{\protect \mathrm {eq}}$ implies $0<\protect
  \mathaccentV {tilde}07E{T}<1$.}\BibitemShut {Stop}%
\bibitem [{Note2()}]{Note2}%
  \BibitemOpen
  \bibinfo {note} {The name comes from the fact that $\protect \mathcal
  {U}(\protect \mathbf {q})$ delivers the mean force, i.e. $-\nabla _{\bq } \mU
  (\bq )=-\langle \nabla _{\bx }U(\bx )\delta (\bgam (\bx )-\bq )\rangle
  $.}\BibitemShut {Stop}%
\bibitem [{\citenamefont {Kirkwood}(1935)}]{Kirkwood}%
  \BibitemOpen
  \bibfield  {author} {\bibinfo {author} {\bibfnamefont {J.~G.}\ \bibnamefont
  {Kirkwood}},\ }\href {\doibase 10.1063/1.1749657} {\bibfield  {journal}
  {\bibinfo  {journal} {J. Chem. Phys.}\ }\textbf {\bibinfo {volume} {3}},\
  \bibinfo {pages} {300–313} (\bibinfo {year} {1935})}\BibitemShut {NoStop}%
\bibitem [{\citenamefont {Yang}\ \emph {et~al.}(2002)\citenamefont {Yang},
  \citenamefont {Witkoskie},\ and\ \citenamefont
  {Cao}}]{yang_single-molecule_2002}%
  \BibitemOpen
  \bibfield  {author} {\bibinfo {author} {\bibfnamefont {S.}~\bibnamefont
  {Yang}}, \bibinfo {author} {\bibfnamefont {J.~B.}\ \bibnamefont {Witkoskie}},
  \ and\ \bibinfo {author} {\bibfnamefont {J.}~\bibnamefont {Cao}},\ }\href
  {\doibase 10.1063/1.1521156} {\bibfield  {journal} {\bibinfo  {journal} {J.
  Chem. Phys.}\ }\textbf {\bibinfo {volume} {117}},\ \bibinfo {pages} {11010}
  (\bibinfo {year} {2002})}\BibitemShut {NoStop}%
\bibitem [{\citenamefont {Joo}\ \emph {et~al.}(2008)\citenamefont {Joo},
  \citenamefont {Balci}, \citenamefont {Ishitsuka}, \citenamefont
  {Buranachai},\ and\ \citenamefont {Ha}}]{joo_advances_2008}%
  \BibitemOpen
  \bibfield  {author} {\bibinfo {author} {\bibfnamefont {C.}~\bibnamefont
  {Joo}}, \bibinfo {author} {\bibfnamefont {H.}~\bibnamefont {Balci}}, \bibinfo
  {author} {\bibfnamefont {Y.}~\bibnamefont {Ishitsuka}}, \bibinfo {author}
  {\bibfnamefont {C.}~\bibnamefont {Buranachai}}, \ and\ \bibinfo {author}
  {\bibfnamefont {T.}~\bibnamefont {Ha}},\ }\href {\doibase
  10.1146/annurev.biochem.77.070606.101543} {\bibfield  {journal} {\bibinfo
  {journal} {Annu. Rev. Biochem.}\ }\textbf {\bibinfo {volume} {77}},\ \bibinfo
  {pages} {51} (\bibinfo {year} {2008})}\BibitemShut {NoStop}%
\bibitem [{\citenamefont {Moro}(1995)}]{Moro}%
  \BibitemOpen
  \bibfield  {author} {\bibinfo {author} {\bibfnamefont {G.~J.}\ \bibnamefont
  {Moro}},\ }\href {\doibase 10.1063/1.470320} {\bibfield  {journal} {\bibinfo
  {journal} {J. Chem. Phys.}\ }\textbf {\bibinfo {volume} {103}},\ \bibinfo
  {pages} {7514–7531} (\bibinfo {year} {1995})}\BibitemShut {NoStop}%
\bibitem [{\citenamefont {Schmiedl}\ and\ \citenamefont
  {Seifert}(2007)}]{Schmiedl}%
  \BibitemOpen
  \bibfield  {author} {\bibinfo {author} {\bibfnamefont {T.}~\bibnamefont
  {Schmiedl}}\ and\ \bibinfo {author} {\bibfnamefont {U.}~\bibnamefont
  {Seifert}},\ }\href {\doibase 10.1209/0295-5075/81/20003} {\bibfield
  {journal} {\bibinfo  {journal} {{EPL}}\ }\textbf {\bibinfo {volume} {81}},\
  \bibinfo {pages} {20003} (\bibinfo {year} {2007})}\BibitemShut {NoStop}%
\bibitem [{\citenamefont {Ouerdane}\ \emph {et~al.}(2015)\citenamefont
  {Ouerdane}, \citenamefont {Apertet}, \citenamefont {Goupil},\ and\
  \citenamefont {Lecoeur}}]{Ouerdane}%
  \BibitemOpen
  \bibfield  {author} {\bibinfo {author} {\bibfnamefont {H.}~\bibnamefont
  {Ouerdane}}, \bibinfo {author} {\bibfnamefont {Y.}~\bibnamefont {Apertet}},
  \bibinfo {author} {\bibfnamefont {C.}~\bibnamefont {Goupil}}, \ and\ \bibinfo
  {author} {\bibfnamefont {P.}~\bibnamefont {Lecoeur}},\ }\href {\doibase
  10.1140/epjst/e2015-02431-x} {\bibfield  {journal} {\bibinfo  {journal} {Eur.
  Phys. J. Spec. Top.}\ }\textbf {\bibinfo {volume} {224}},\ \bibinfo {pages}
  {839–864} (\bibinfo {year} {2015})}\BibitemShut {NoStop}%
\end{thebibliography}

\begin{thebibliography}{4}%
\makeatletter
\providecommand \@ifxundefined [1]{%
 \@ifx{#1\undefined}
}%
\providecommand \@ifnum [1]{%
 \ifnum #1\expandafter \@firstoftwo
 \else \expandafter \@secondoftwo
 \fi
}%
\providecommand \@ifx [1]{%
 \ifx #1\expandafter \@firstoftwo
 \else \expandafter \@secondoftwo
 \fi
}%
\providecommand \natexlab [1]{#1}%
\providecommand \enquote  [1]{``#1''}%
\providecommand \bibnamefont  [1]{#1}%
\providecommand \bibfnamefont [1]{#1}%
\providecommand \citenamefont [1]{#1}%
\providecommand \href@noop [0]{\@secondoftwo}%
\providecommand \href [0]{\begingroup \@sanitize@url \@href}%
\providecommand \@href[1]{\@@startlink{#1}\@@href}%
\providecommand \@@href[1]{\endgroup#1\@@endlink}%
\providecommand \@sanitize@url [0]{\catcode `\\12\catcode `\$12\catcode
  `\&12\catcode `\#12\catcode `\^12\catcode `\_12\catcode `\%12\relax}%
\providecommand \@@startlink[1]{}%
\providecommand \@@endlink[0]{}%
\providecommand \url  [0]{\begingroup\@sanitize@url \@url }%
\providecommand \@url [1]{\endgroup\@href {#1}{\urlprefix }}%
\providecommand \urlprefix  [0]{URL }%
\providecommand \Eprint [0]{\href }%
\providecommand \doibase [0]{https://doi.org/}%
\providecommand \selectlanguage [0]{\@gobble}%
\providecommand \bibinfo  [0]{\@secondoftwo}%
\providecommand \bibfield  [0]{\@secondoftwo}%
\providecommand \translation [1]{[#1]}%
\providecommand \BibitemOpen [0]{}%
\providecommand \bibitemStop [0]{}%
\providecommand \bibitemNoStop [0]{.\EOS\space}%
\providecommand \EOS [0]{\spacefactor3000\relax}%
\providecommand \BibitemShut  [1]{\csname bibitem#1\endcsname}%
\let\auto@bib@innerbib\@empty
\bibitem [{\citenamefont {Lapolla}\ and\ \citenamefont {Godec}(2020)}]{PRL}%
  \BibitemOpen
  \bibfield  {author} {\bibinfo {author} {\bibfnamefont {A.}~\bibnamefont
  {Lapolla}}\ and\ \bibinfo {author} {\bibfnamefont {A.}~\bibnamefont
  {Godec}},\ }\href {https://doi.org/10.1103/PhysRevLett.125.110602} {\bibfield
   {journal} {\bibinfo  {journal} {Phys. Rev. Lett.}\ }\textbf {\bibinfo
  {volume} {125}},\ \bibinfo {pages} {110602} (\bibinfo {year}
  {2020})}\BibitemShut {NoStop}%
\bibitem [{Note1()}]{Note1}%
  \BibitemOpen
  \bibinfo {note} {There are also two inessential factors missing, which,
  however, have no bearing.}\BibitemShut {Stop}%
\bibitem [{\citenamefont {Tops{\o}e}(2007)}]{log}%
  \BibitemOpen
  \bibfield  {author} {\bibinfo {author} {\bibfnamefont {F.}~\bibnamefont
  {Tops{\o}e}},\ }\bibinfo {title} {Some bounds for the logarithmic function},\
  in\ \href@noop {} {\emph {\bibinfo {booktitle} {Inequality Theory and
  Applications}}},\ Vol.~\bibinfo {volume} {4},\ \bibinfo {editor} {edited by\
  \bibinfo {editor} {\bibfnamefont {Y.}~\bibnamefont {Cho}}, \bibinfo {editor}
  {\bibfnamefont {J.}~\bibnamefont {Kim}},\ and\ \bibinfo {editor}
  {\bibfnamefont {S.}~\bibnamefont {Dragomir}}}\ (\bibinfo  {publisher} {Nova
  Science Publishers},\ \bibinfo {address} {United States},\ \bibinfo {year}
  {2007})\ pp.\ \bibinfo {pages} {137--151}\BibitemShut {NoStop}%
\bibitem [{\citenamefont {Stewart}(2009)}]{Lambert}%
  \BibitemOpen
  \bibfield  {author} {\bibinfo {author} {\bibfnamefont {S.~M.}\ \bibnamefont
  {Stewart}},\ }\href {http://eudml.org/doc/224588} {\bibfield  {journal}
  {\bibinfo  {journal} {J. Inequal. Pure Appl. Math.}\ }\textbf {\bibinfo
  {volume} {10}},\ \bibinfo {pages} {Paper No. 96, 4 p., electronic only}
  (\bibinfo {year} {2009})}\BibitemShut {NoStop}%
\end{thebibliography}

\begin{thebibliography}{16}
\makeatletter
\providecommand \@ifxundefined [1]{%
 \@ifx{#1\undefined}
}%
\providecommand \@ifnum [1]{%
 \ifnum #1\expandafter \@firstoftwo
 \else \expandafter \@secondoftwo
 \fi
}%
\providecommand \@ifx [1]{%
 \ifx #1\expandafter \@firstoftwo
 \else \expandafter \@secondoftwo
 \fi
}%
\providecommand \natexlab [1]{#1}%
\providecommand \enquote  [1]{``#1''}%
\providecommand \bibnamefont  [1]{#1}%
\providecommand \bibfnamefont [1]{#1}%
\providecommand \citenamefont [1]{#1}%
\providecommand \href@noop [0]{\@secondoftwo}%
\providecommand \href [0]{\begingroup \@sanitize@url \@href}%
\providecommand \@href[1]{\@@startlink{#1}\@@href}%
\providecommand \@@href[1]{\endgroup#1\@@endlink}%
\providecommand \@sanitize@url [0]{\catcode `\\12\catcode `\$12\catcode
  `\&12\catcode `\#12\catcode `\^12\catcode `\_12\catcode `\%12\relax}%
\providecommand \@@startlink[1]{}%
\providecommand \@@endlink[0]{}%
\providecommand \url  [0]{\begingroup\@sanitize@url \@url }%
\providecommand \@url [1]{\endgroup\@href {#1}{\urlprefix }}%
\providecommand \urlprefix  [0]{URL }%
\providecommand \Eprint [0]{\href }%
\providecommand \doibase [0]{http://dx.doi.org/}%
\providecommand \selectlanguage [0]{\@gobble}%
\providecommand \bibinfo  [0]{\@secondoftwo}%
\providecommand \bibfield  [0]{\@secondoftwo}%
\providecommand \translation [1]{[#1]}%
\providecommand \BibitemOpen [0]{}%
\providecommand \bibitemStop [0]{}%
\providecommand \bibitemNoStop [0]{.\EOS\space}%
\providecommand \EOS [0]{\spacefactor3000\relax}%
\providecommand \BibitemShut  [1]{\csname bibitem#1\endcsname}%
\let\auto@bib@innerbib\@empty
\bibitem [{\citenamefont {Chatzigeorgiou}(2013)}]{SWbound}%
  \BibitemOpen
  \bibfield  {author} {\bibinfo {author} {\bibfnamefont {I.}~\bibnamefont
  {Chatzigeorgiou}},\ }\href {\doibase 10.1109/lcomm.2013.070113.130972}
  {\bibfield  {journal} {\bibinfo  {journal} {IEEE Communications Letters}\
  }\textbf {\bibinfo {volume} {17}},\ \bibinfo {pages} {1505–1508} (\bibinfo
  {year} {2013})}\BibitemShut {NoStop}%
\bibitem [{\citenamefont {Lizana}\ and\ \citenamefont
  {Ambjörnsson}(2009)}]{Slizana_diffusion_2009}%
  \BibitemOpen
  \bibfield  {author} {\bibinfo {author} {\bibfnamefont {L.}~\bibnamefont
  {Lizana}}\ and\ \bibinfo {author} {\bibfnamefont {T.}~\bibnamefont
  {Ambjörnsson}},\ }\href {\doibase 10.1103/PhysRevE.80.051103} {\bibfield
  {journal} {\bibinfo  {journal} {Phys. Rev. E}\ }\textbf {\bibinfo {volume}
  {80}},\ \bibinfo {pages} {051103} (\bibinfo {year} {2009})}\BibitemShut
  {NoStop}%
\bibitem [{\citenamefont {Lizana}\ and\ \citenamefont
  {Ambjörnsson}(2008)}]{Slizana_single-file_2008}%
  \BibitemOpen
  \bibfield  {author} {\bibinfo {author} {\bibfnamefont {L.}~\bibnamefont
  {Lizana}}\ and\ \bibinfo {author} {\bibfnamefont {T.}~\bibnamefont
  {Ambjörnsson}},\ }\href {\doibase 10.1103/PhysRevLett.100.200601} {\bibfield
   {journal} {\bibinfo  {journal} {Phys. Rev. Lett.}\ }\textbf {\bibinfo
  {volume} {100}},\ \bibinfo {pages} {200601} (\bibinfo {year}
  {2008})}\BibitemShut {NoStop}%
\bibitem [{\citenamefont {Lapolla}\ and\ \citenamefont
  {Godec}(2019)}]{Slapolla_manifestations_2019}%
  \BibitemOpen
  \bibfield  {author} {\bibinfo {author} {\bibfnamefont {A.}~\bibnamefont
  {Lapolla}}\ and\ \bibinfo {author} {\bibfnamefont {A.}~\bibnamefont
  {Godec}},\ }\href {\doibase 10.3389/fphy.2019.00182} {\bibfield  {journal}
  {\bibinfo  {journal} {Front. Phys.}\ }\textbf {\bibinfo {volume} {7}}
  (\bibinfo {year} {2019}),\ 10.3389/fphy.2019.00182}\BibitemShut {NoStop}%
\bibitem [{\citenamefont {Godec}\ \emph {et~al.}(2010)\citenamefont {Godec},
  \citenamefont {Ukmar}, \citenamefont {Gaber{\v{s}}{\v{c}}ek},\ and\
  \citenamefont {Merzel}}]{SAG_ADI}%
  \BibitemOpen
  \bibfield  {author} {\bibinfo {author} {\bibfnamefont {A.}~\bibnamefont
  {Godec}}, \bibinfo {author} {\bibfnamefont {T.}~\bibnamefont {Ukmar}},
  \bibinfo {author} {\bibfnamefont {M.}~\bibnamefont {Gaber{\v{s}}{\v{c}}ek}},
  \ and\ \bibinfo {author} {\bibfnamefont {F.}~\bibnamefont {Merzel}},\ }\href
  {\doibase 10.1209/0295-5075/92/60011} {\bibfield  {journal} {\bibinfo
  {journal} {{EPL} (Europhysics Letters)}\ }\textbf {\bibinfo {volume} {92}},\
  \bibinfo {pages} {60011} (\bibinfo {year} {2010})}\BibitemShut {NoStop}%
\bibitem [{\citenamefont {Moro}(1995)}]{SMoro}%
  \BibitemOpen
  \bibfield  {author} {\bibinfo {author} {\bibfnamefont {G.~J.}\ \bibnamefont
  {Moro}},\ }\href {\doibase 10.1063/1.470320} {\bibfield  {journal} {\bibinfo
  {journal} {J. Chem. Phys.}\ }\textbf {\bibinfo {volume} {103}},\ \bibinfo
  {pages} {7514–7531} (\bibinfo {year} {1995})}\BibitemShut {NoStop}%
\bibitem [{\citenamefont {{E.B. Mpemba D.G.
  Osborne}}(1979)}]{Se.b._mpemba_d.g._osborne_cool?_1979}%
  \BibitemOpen
  \bibfield  {author} {\bibinfo {author} {\bibnamefont {{E.B. Mpemba D.G.
  Osborne}}},\ }\href {\doibase 10.1088/0031-9120/14/7/312} {\bibfield
  {journal} {\bibinfo  {journal} {Physics Education}\ }\textbf {\bibinfo
  {volume} {14}},\ \bibinfo {pages} {410} (\bibinfo {year} {1979})}\BibitemShut
  {NoStop}%
\bibitem [{\citenamefont {Jeng}(2006)}]{Sjeng_mpemba_2006}%
  \BibitemOpen
  \bibfield  {author} {\bibinfo {author} {\bibfnamefont {M.}~\bibnamefont
  {Jeng}},\ }\href {\doibase 10.1119/1.2186331} {\bibfield  {journal} {\bibinfo
   {journal} {American Journal of Physics}\ }\textbf {\bibinfo {volume} {74}},\
  \bibinfo {pages} {514} (\bibinfo {year} {2006})}\BibitemShut {NoStop}%
\bibitem [{\citenamefont {Katz}(2009)}]{Skatz_when_2009}%
  \BibitemOpen
  \bibfield  {author} {\bibinfo {author} {\bibfnamefont {J.~I.}\ \bibnamefont
  {Katz}},\ }\href {\doibase 10.1119/1.2996187} {\bibfield  {journal} {\bibinfo
   {journal} {American Journal of Physics}\ }\textbf {\bibinfo {volume} {77}},\
  \bibinfo {pages} {27} (\bibinfo {year} {2009})}\BibitemShut {NoStop}%
\bibitem [{\citenamefont {Chaddah}\ \emph {et~al.}(2010)\citenamefont
  {Chaddah}, \citenamefont {Dash}, \citenamefont {Kumar},\ and\ \citenamefont
  {Banerjee}}]{Schaddah_overtaking_2010}%
  \BibitemOpen
  \bibfield  {author} {\bibinfo {author} {\bibfnamefont {P.}~\bibnamefont
  {Chaddah}}, \bibinfo {author} {\bibfnamefont {S.}~\bibnamefont {Dash}},
  \bibinfo {author} {\bibfnamefont {K.}~\bibnamefont {Kumar}}, \ and\ \bibinfo
  {author} {\bibfnamefont {A.}~\bibnamefont {Banerjee}},\ }\href
  {http://arxiv.org/abs/1011.3598} {\bibfield  {journal} {\bibinfo  {journal}
  {arXiv:1011.3598 [cond-mat, physics:physics]}\ } (\bibinfo {year} {2010})},\
  \bibinfo {note} {arXiv: 1011.3598}\BibitemShut {NoStop}%
\bibitem [{\citenamefont {Greaney}\ \emph {et~al.}(2011)\citenamefont
  {Greaney}, \citenamefont {Lani}, \citenamefont {Cicero},\ and\ \citenamefont
  {Grossman}}]{Sgreaney_mpemba-like_2011}%
  \BibitemOpen
  \bibfield  {author} {\bibinfo {author} {\bibfnamefont {P.~A.}\ \bibnamefont
  {Greaney}}, \bibinfo {author} {\bibfnamefont {G.}~\bibnamefont {Lani}},
  \bibinfo {author} {\bibfnamefont {G.}~\bibnamefont {Cicero}}, \ and\ \bibinfo
  {author} {\bibfnamefont {J.~C.}\ \bibnamefont {Grossman}},\ }\href {\doibase
  10.1007/s11661-011-0843-4} {\bibfield  {journal} {\bibinfo  {journal} {Metall
  and Mat Trans A}\ }\textbf {\bibinfo {volume} {42}},\ \bibinfo {pages} {3907}
  (\bibinfo {year} {2011})}\BibitemShut {NoStop}%
\bibitem [{\citenamefont {Hu}\ \emph {et~al.}(2018)\citenamefont {Hu},
  \citenamefont {Li}, \citenamefont {Huang}, \citenamefont {Li}, \citenamefont
  {Luo}, \citenamefont {Chen}, \citenamefont {Jiang},\ and\ \citenamefont
  {An}}]{Shu_conformation_2018}%
  \BibitemOpen
  \bibfield  {author} {\bibinfo {author} {\bibfnamefont {C.}~\bibnamefont
  {Hu}}, \bibinfo {author} {\bibfnamefont {J.}~\bibnamefont {Li}}, \bibinfo
  {author} {\bibfnamefont {S.}~\bibnamefont {Huang}}, \bibinfo {author}
  {\bibfnamefont {H.}~\bibnamefont {Li}}, \bibinfo {author} {\bibfnamefont
  {C.}~\bibnamefont {Luo}}, \bibinfo {author} {\bibfnamefont {J.}~\bibnamefont
  {Chen}}, \bibinfo {author} {\bibfnamefont {S.}~\bibnamefont {Jiang}}, \ and\
  \bibinfo {author} {\bibfnamefont {L.}~\bibnamefont {An}},\ }\href {\doibase
  10.1021/acs.cgd.8b01250} {\bibfield  {journal} {\bibinfo  {journal} {Crystal
  Growth \& Design}\ }\textbf {\bibinfo {volume} {18}},\ \bibinfo {pages}
  {5757} (\bibinfo {year} {2018})}\BibitemShut {NoStop}%
\bibitem [{\citenamefont {Ahn}\ \emph {et~al.}(2016)\citenamefont {Ahn},
  \citenamefont {Kang}, \citenamefont {Koh},\ and\ \citenamefont
  {Lee}}]{Sahn_experimental_2016}%
  \BibitemOpen
  \bibfield  {author} {\bibinfo {author} {\bibfnamefont {Y.-H.}\ \bibnamefont
  {Ahn}}, \bibinfo {author} {\bibfnamefont {H.}~\bibnamefont {Kang}}, \bibinfo
  {author} {\bibfnamefont {D.-Y.}\ \bibnamefont {Koh}}, \ and\ \bibinfo
  {author} {\bibfnamefont {H.}~\bibnamefont {Lee}},\ }\href {\doibase
  10.1007/s11814-016-0029-2} {\bibfield  {journal} {\bibinfo  {journal} {Korean
  J. Chem. Eng.}\ }\textbf {\bibinfo {volume} {33}},\ \bibinfo {pages} {1903}
  (\bibinfo {year} {2016})}\BibitemShut {NoStop}%
\bibitem [{\citenamefont {Lasanta}\ \emph {et~al.}(2017)\citenamefont
  {Lasanta}, \citenamefont {Vega~Reyes}, \citenamefont {Prados},\ and\
  \citenamefont {Santos}}]{Slasanta_when_2017}%
  \BibitemOpen
  \bibfield  {author} {\bibinfo {author} {\bibfnamefont {A.}~\bibnamefont
  {Lasanta}}, \bibinfo {author} {\bibfnamefont {F.}~\bibnamefont {Vega~Reyes}},
  \bibinfo {author} {\bibfnamefont {A.}~\bibnamefont {Prados}}, \ and\ \bibinfo
  {author} {\bibfnamefont {A.}~\bibnamefont {Santos}},\ }\href {\doibase
  10.1103/PhysRevLett.119.148001} {\bibfield  {journal} {\bibinfo  {journal}
  {Phys. Rev. Lett.}\ }\textbf {\bibinfo {volume} {119}},\ \bibinfo {pages}
  {148001} (\bibinfo {year} {2017})}\BibitemShut {NoStop}%
\bibitem [{\citenamefont {collaboration}\ \emph {et~al.}(2019)\citenamefont
  {collaboration}, \citenamefont {Baity-Jesi}, \citenamefont {Calore},
  \citenamefont {Cruz}, \citenamefont {Fernandez}, \citenamefont {Gil-Narvion},
  \citenamefont {Gordillo-Guerrero}, \citenamefont {Iñiguez}, \citenamefont
  {Lasanta}, \citenamefont {Maiorano}, \citenamefont {Marinari}, \citenamefont
  {Martin-Mayor}, \citenamefont {Moreno-Gordo}, \citenamefont {Muñoz-Sudupe},
  \citenamefont {Navarro}, \citenamefont {Parisi}, \citenamefont
  {Perez-Gaviro}, \citenamefont {Ricci-Tersenghi}, \citenamefont
  {Ruiz-Lorenzo}, \citenamefont {Schifano}, \citenamefont {Seoane},
  \citenamefont {Tarancon}, \citenamefont {Tripiccione},\ and\ \citenamefont
  {Yllanes}}]{Sjanus_collaboration_mpemba_2019}%
  \BibitemOpen
  \bibfield  {author} {\bibinfo {author} {\bibfnamefont {J.}~\bibnamefont
  {collaboration}}, \bibinfo {author} {\bibfnamefont {M.}~\bibnamefont
  {Baity-Jesi}}, \bibinfo {author} {\bibfnamefont {E.}~\bibnamefont {Calore}},
  \bibinfo {author} {\bibfnamefont {A.}~\bibnamefont {Cruz}}, \bibinfo {author}
  {\bibfnamefont {L.~A.}\ \bibnamefont {Fernandez}}, \bibinfo {author}
  {\bibfnamefont {J.~M.}\ \bibnamefont {Gil-Narvion}}, \bibinfo {author}
  {\bibfnamefont {A.}~\bibnamefont {Gordillo-Guerrero}}, \bibinfo {author}
  {\bibfnamefont {D.}~\bibnamefont {Iñiguez}}, \bibinfo {author}
  {\bibfnamefont {A.}~\bibnamefont {Lasanta}}, \bibinfo {author} {\bibfnamefont
  {A.}~\bibnamefont {Maiorano}}, \bibinfo {author} {\bibfnamefont
  {E.}~\bibnamefont {Marinari}}, \bibinfo {author} {\bibfnamefont
  {V.}~\bibnamefont {Martin-Mayor}}, \bibinfo {author} {\bibfnamefont
  {J.}~\bibnamefont {Moreno-Gordo}}, \bibinfo {author} {\bibfnamefont
  {A.}~\bibnamefont {Muñoz-Sudupe}}, \bibinfo {author} {\bibfnamefont
  {D.}~\bibnamefont {Navarro}}, \bibinfo {author} {\bibfnamefont
  {G.}~\bibnamefont {Parisi}}, \bibinfo {author} {\bibfnamefont
  {S.}~\bibnamefont {Perez-Gaviro}}, \bibinfo {author} {\bibfnamefont
  {F.}~\bibnamefont {Ricci-Tersenghi}}, \bibinfo {author} {\bibfnamefont
  {J.~J.}\ \bibnamefont {Ruiz-Lorenzo}}, \bibinfo {author} {\bibfnamefont
  {S.~F.}\ \bibnamefont {Schifano}}, \bibinfo {author} {\bibfnamefont
  {B.}~\bibnamefont {Seoane}}, \bibinfo {author} {\bibfnamefont
  {A.}~\bibnamefont {Tarancon}}, \bibinfo {author} {\bibfnamefont
  {R.}~\bibnamefont {Tripiccione}}, \ and\ \bibinfo {author} {\bibfnamefont
  {D.}~\bibnamefont {Yllanes}},\ }\href {\doibase 10.1073/pnas.1819803116}
  {\bibfield  {journal} {\bibinfo  {journal} {Proc Natl Acad Sci USA}\ }\textbf
  {\bibinfo {volume} {116}},\ \bibinfo {pages} {15350} (\bibinfo {year}
  {2019})},\ \bibinfo {note} {arXiv: 1804.07569}\BibitemShut {NoStop}%
\bibitem [{\citenamefont {Lu}\ and\ \citenamefont
  {Raz}(2017)}]{Slu_nonequilibrium_2017}%
  \BibitemOpen
  \bibfield  {author} {\bibinfo {author} {\bibfnamefont {Z.}~\bibnamefont
  {Lu}}\ and\ \bibinfo {author} {\bibfnamefont {O.}~\bibnamefont {Raz}},\
  }\href {\doibase 10.1073/pnas.1701264114} {\bibfield  {journal} {\bibinfo
  {journal} {Proc Natl Acad Sci USA}\ }\textbf {\bibinfo {volume} {114}},\
  \bibinfo {pages} {5083} (\bibinfo {year} {2017})}\BibitemShut {NoStop}%
\bibitem [{\citenamefont {Klich}\ and\ \citenamefont
  {Vucelja}(2019)}]{Sklich_solution_2019}%
  \BibitemOpen
  \bibfield  {author} {\bibinfo {author} {\bibfnamefont {I.}~\bibnamefont
  {Klich}}\ and\ \bibinfo {author} {\bibfnamefont {M.}~\bibnamefont
  {Vucelja}},\ }\href {http://arxiv.org/abs/1812.11962} {\bibfield  {journal}
  {\bibinfo  {journal} {arXiv:1812.11962 [cond-mat, physics:math-ph]}\ }
  (\bibinfo {year} {2019})},\ \bibinfo {note} {arXiv: 1812.11962}\BibitemShut
  {NoStop}%
\bibitem [{\citenamefont {Klich}\ \emph {et~al.}(2019)\citenamefont {Klich},
  \citenamefont {Raz}, \citenamefont {Hirschberg},\ and\ \citenamefont
  {Vucelja}}]{Sklich_mpemba_2019}%
  \BibitemOpen
  \bibfield  {author} {\bibinfo {author} {\bibfnamefont {I.}~\bibnamefont
  {Klich}}, \bibinfo {author} {\bibfnamefont {O.}~\bibnamefont {Raz}}, \bibinfo
  {author} {\bibfnamefont {O.}~\bibnamefont {Hirschberg}}, \ and\ \bibinfo
  {author} {\bibfnamefont {M.}~\bibnamefont {Vucelja}},\ }\href {\doibase
  10.1103/PhysRevX.9.021060} {\bibfield  {journal} {\bibinfo  {journal} {Phys.
  Rev. X}\ }\textbf {\bibinfo {volume} {9}},\ \bibinfo {pages} {021060}
  (\bibinfo {year} {2019})}\BibitemShut {NoStop}%
      \end{thebibliography}
      
\end{document}